\documentclass[journal, comsoc, onecolumn, 12pt]{IEEEtran}

\usepackage{multirow}
\usepackage{multicol}
\usepackage{graphicx}
\usepackage{comment}
\usepackage{amsmath}
\usepackage{caption}
\usepackage{makecell}
\usepackage{amsthm}
\usepackage{array}
\newcolumntype{P}[1]{>{\centering\arraybackslash}p{#1}}
\newtheorem{theorem}{Theorem}

\newtheorem{lemma}{Lemma}

\DeclareMathOperator*{\argmin}{arg\,min}
\makeatletter

\providecommand{\examplename}{Example}
\makeatother

\usepackage{amsfonts}
\usepackage{stackengine}
\usepackage{epstopdf}
\usepackage[]{algorithm}
\usepackage{algorithmic}
\usepackage{setspace,soul}
\usepackage{caption}
\ifCLASSINFOpdf
\else
\fi
\usepackage{amsmath}
\usepackage[cmintegrals]{newtxmath}
\begin{document}
%
\title{On the Reusability of Post-Experimental Field Data for Underwater Acoustic Communications R\&D}
%
%
%

\author{Sijung~Yang$^{1}$,~Grant~Deane$^2$,~James~C.~Preisig$^3$, ~Noyan~C.~Sev\"uktekin$^{1}$,~Jae~W.~Choi$^{1}$,~and~Andrew~C.~Singer$^{1}$,~\IEEEmembership{Fellow,~IEEE}
\thanks{$^1$Coordinated Science Laboratory, Department of Electrical and Computer Engineering, University of Illinois at Urbana-Champaign, Uraban, IL USA.}
\thanks{$^2$Scripps Institution of Oceanography, University of California at San Diego, La Jolla, CA USA.}
\thanks{$^3$JPAnalytics LLC, Falmouth, MA USA.}
\thanks{Corresponding author: acsinger@illinois.edu.}
}%

%
%

\markboth{IEEE Journal of Oceanic Engineering - for submission}%
{Yang \MakeLowercase{\textit{et al.}}: On the Reusability of Post-Experimental Field Data for Underwater Acoustic Communications R\&D}
%



\maketitle

\begin{abstract}
Field data is often expensive to collect, time-consuming to prepare to collect, and even more time-consuming to process after the experiment has concluded. However, it is often the practice that such data are used for little after the funded research activity that was concomitant with the experiment is completed. Immutability of the original experimental configuration either results in re-gathering of expensive field-data, or in absence of such data, model-dependent analysis that partially captures the real-world dynamics. For underwater acoustic research and development, the standard communication pipeline might be modified to enable greater re-usability of experimental field data. This paper first characterizes the necessary modifications to the standard communication pipeline to prepare signals for transmission and subsequent recording such that research trades for different modulation and coding schemes may be undertaken post-experiment, without the need for re-transmission of additional waveforms. Then, using the modified mathematical framework, sufficient conditions for reliable post-experimental replay of the environment are recognized. Finally, techniques are discussed to collect sufficient environmental statistics such that subsequent research can be accomplished long after the experiment has been completed, and that results from a given experiment may be reasonably compared with those of another. Examples are provided using both synthetic and experimental data collected from at-sea field tests.
\end{abstract}


%
\IEEEpeerreviewmaketitle

\section{Introduction}
\IEEEPARstart{U}{nderwater} acoustic communication systems have been developed and tested experimentally in a wide range of environments. While it is generally accepted among researchers and experimentalists that the doubly spread nature of the underwater acoustic channel makes it one of the most challenging communications environments on the planet, it receives comparatively little attention in the research literature for a number of reasons.

One of these has to do with the cost of deploying and testing different algorithmic methodologies in real environments, and another, that goes hand-in-hand with the first, is that there is no commonly accepted channel model through which performance can be meaningfully validated. While there continues to be strong interest in the scientific and Naval communities in the technology in general, by comparison, RF-wireless and wireline communication systems, which have consumer-facing applications, garner orders-of-magnitude more funding for research and commercial development. 

It is the latter that has led to the widespread adoption of mathematically convenient, if not realistic, propagation models that enable simulation and numerical performance evaluation without the need for expensive real-world system deployment. Without readily-available performance evaluation models or compendia of commercially collected experimental data, common practice in the research community has been to make best use of experimental opportunities in real-ocean environments as much as possible; ``in the absence of good statistical models for simulation, experimental demonstration of candidate communication schemes remains a \textit{de facto} standard''
\cite{stojanovic2009underwater}.
In preparing and planning for an experiment, much expense and preparation are needed, and while great pains are often taken to ensure the fidelity of the collected data for the original goals of the experiment, comparatively little attention is paid to ensuring that subsequent developments in physical-layer algorithms, modulation and coding concepts can be explored with the collected data.  Some of the best practices that have been or can be employed in testing and experimental data collection activities are described to make such subsequent development possible. 
The Workshop Report from the UComms12 Conference \cite{potter2013ucomms} presented many of the best practices for both acoustic and environmental data collection during field experiments to facilitate the sharing and comparison of data and results across field experiments. Yet, today it is still difficult to meaningfully compare results across experiments and high quality data that enables the direct comparison of signalling techniques and processing algorithms in a common environment remains scarce. In addition, versatile data sets that enable researchers to develop and validate signals and algorithms in realistic settings continue to be scarce.

Challenges abound and stem, in part, from the apparent strengths and capabilities of data generation techniques and the requirements of different types of research.  Physics-based numerical simulators are some of the most versatile tools available for varying an experimental environment but remain the least realistic.  Field tests are by definition the most realistic technique for collecting data for experimental validation, and are not only the least versatile, but can also be prohibitively expensive. All the signals to be considered must be transmitted during the tests and it is often difficult to accurately measure all of the environmental conditions that may impact signal or algorithm performance.  Playback simulators offer a compromise approach, but are limited by the accuracy of the required channel parameterization, including the ambient noise modeling, and estimates of these quantities. As a result, they often provide overly optimistic performance predictions that do not hold up in practice.

Here, three techniques to overcome some of these challenges are presented. First, we show that post-experiment performance evaluations with different forward error correction (FEC) schemes and different symbol mappings can be enabled, without the need for conducting additional field experiments or relying on playback simulators having accurate channel estimates. This can be done by virtually inserting arbitrary dither sequences into the signal preparation stage, which only changes the meaning of the transmitted signal, and not the waveform itself. Through the use of binary dithers (at the binary sequence level), arbitrary error correction coding and coded-modulation modifications can be made post-experiment, without any assumptions on the environment or experimental conditions, so long as the transmitted symbol sequence remains unchanged. Through the use of additive dithers at the symbol constellation level, the performance of a variety of transmit modulation schemes can be explored post experiment, where the fidelity of such performance predictions rely on the linearity of the end-to-end signal path and suitably successful mitigation of intersymbol interference (ISI) through equalization that incorporates the additive dither. Also, we suggest a modification to direct playback simulation techniques by exploiting prior information of not only an estimated channel response, but also making use of the residual prediction error (RPE) obtained by subtracting the experimentally received signal from that which would be estimated using the estimated channel response. This scheme compensates for overly-optimistic performance prediction results arising from imperfect channel estimates, since RPE includes statistical characteristics of such channel estimation errors that can be injected into the channel replay output. Finally, we introduce some estimation techniques for environmental conditions during field experiments, which can be used in modeling and predicting such statistical behaviors in underwater acoustic channels.

The rest of the paper is organized as follows. In Section II, methods and requirements for constructing transmit signals that allow post-experiment testing of a wide range of signal modulation and coding techniques are presented. In Section III and IV, the proposed post-experiment testing is validated numerically, and experimentally. In Section V, a modification to direct playback simulation techniques that uses RPE is presented. Finally, in Section VI, techniques for estimating some relevant environmental conditions that existed during a field experiment from the recorded acoustic data are presented.

\section{Signal Preparation for Post-experimental Research}
\indent In this section, we explore the use of two additional elements inserted into a relatively standard signal transmission flow to enable a host of post-experiment performance evaluations by varying the use of different forward error correction (FEC) and interleaving schemes, different symbol mapping methods (both recursive for pre-coding and memoryless) and even different transmit signal constellations, all while using a single, unchanged transmit waveform, which, post-experiment, is an immutable result of experiment design choices.

\begin{figure}[h]
	\centering  
	\includegraphics[scale=0.4]{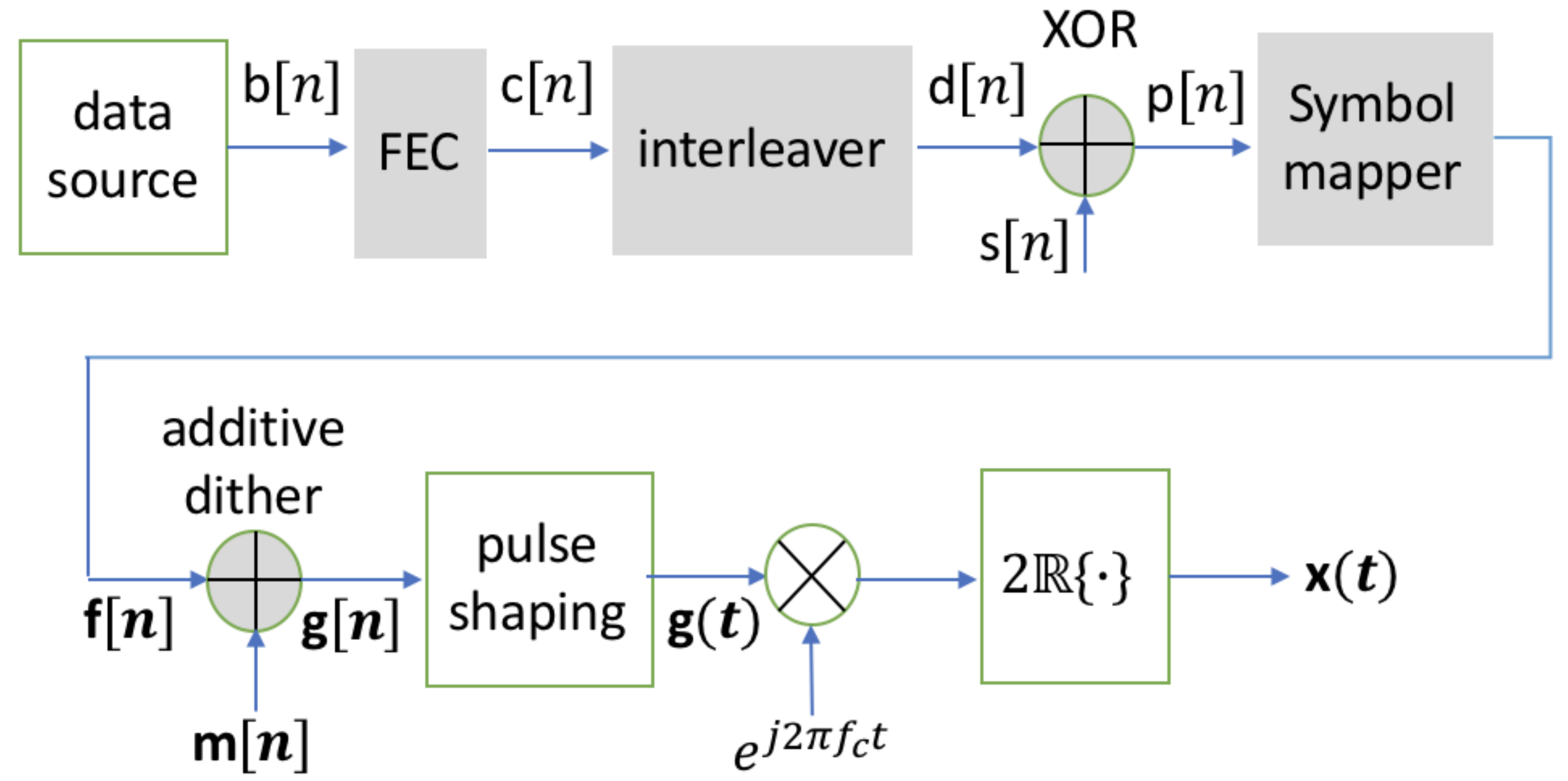}
	\caption{Generic signal flow diagram enabling post-experiment modification of forward error correction, symbol mapping, pre-coding, and constellation evaluation.}
\end{figure}

Our notation is as follows. The transmitted waveform for an experiment $x(t)$, which may be a vector waveform with M channels, i.e. $\boldsymbol{x(t)}=[x_1 (t),\cdots,x_M (t)]^T$, where $x_i (t)$, is transmitted from the $i^{th}$ array element at time $t$. This, as indicated above, is immutable, as it was physically transmitted from the single, or array, transducer during the experiment. However, in the selection of the forward error correction, we have the ability after an experiment to alter the meaning of the data stream, $p[n]$, that is taken as the output of the FEC unit and input to the symbol mapper.  This can be achieved through judicious use of a whitening sequence, $s[n]$, that is applied via XOR to the data stream, $d[n]$, that is the output of the FEC/interleaver stages of encoding.  Whitening sequences are not uncommon, in particular for applications for which, due to framing or other network-layer processing, the data $d[n]$ remains correlated or contains unwanted periodic components due to framing and overhead. This enables better channel utilization, more effective channel coding/decoding performance, and better excitation of the channel, thereby improving identifiability and equalization. 

In our case, however, the sequence $s[n]$ is applied so that after an experiment, the system developer may substitute different FEC and interleaving strategies (i.e. different code rates, different code structures, or even different space-time coding and spreading approaches, if the symbol mapper in Figure 1 produces a vector output $\boldsymbol{f}[n]$). After the experiment, if a different code rate were desired, or a different type of code were to be tested, this can be accommodated using the XOR sequence $s[n]$  as follows. First, the developer can use the same bitstream $b[n]$, to feed into an alternate encoder/interleaver creating a different sequence $d_2 [n]\neq d[n]$.  To produce the same binary sequence $p[n]=d_2[n] \oplus s[n]$, then $s[n]=d_2[n]$  "XNOR" $d[n]$, that is, $s[n]$ is a sequence that takes the value 1 when $d_2[n]=d[n]$ and $0$ when they differ. In the receive signal stream, of Figure 2, the same sequence $s[n]$ can be used to recover an estimate of the desired sequence $d[n]$ for input to the de-interleaver/decoder. A common time index n is used for sequences in Figures 1 and 2 through slight notational abuse. 

If a transmit signal constellation other than the original constellation is desired for testing, then the symbol mapper in Figure 1 can be taken as a lower-rate (or higher-rate) mapper than that which was used in the original experiment, so long as the output sequence desired for testing $\boldsymbol{f_2}[n]$ can be constructed by complex additive dither $\boldsymbol{m}[n]=g[n]-f_2 [n]$. Using this additive dither at both the transmitter and receiver (Figure 2), alternate symbol mapping schemes can be employed, including simple, lower-order, or higher-order, memoryless mappings, like QPSK or a lower-order or higher-order QAM constellation than the one used in the real experiment. When the order of the constellation differs from that of the original transmission, the rate of the transmission will change and transmission signal lengths need to be handled appropriately. This could enable post-experiment trades between use of a higher, more spectrally-efficient transmit constellation, say, using 16QAM and a lower-rate code, say, rate 1/4, and, say, QPSK with a rate 1/2 code. While each has the same net throughput, channel variability and noise conditions may favor one approach. More sophisticated mapping schemes, including trellis-coded modulation, or recursive symbol pre-coding schemes can also be employed. These are attractive for their spectral efficiency improvements \cite{daly2010linear,tuchler2011turbo} as well as their amenability to iterative equalization and decoding methods, such as turbo-equalization. An arbitrarily sophisticated symbol mapper can be employed and tested, so long as the receiver makes use of the same additive dither as the transmitter, without the need to re-transmit data in another experiment.

\begin{figure}[h]
	\centering
	\includegraphics[scale=0.4]{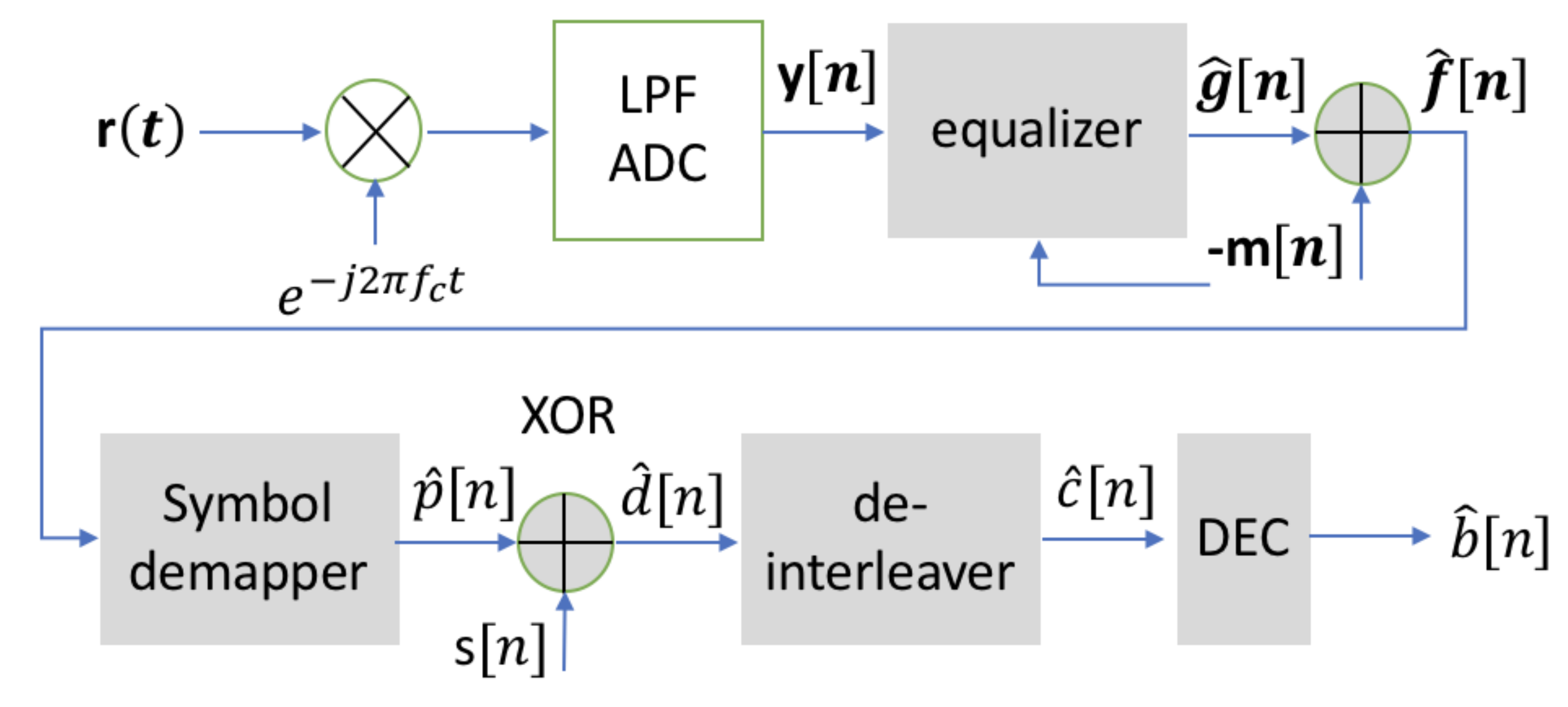}
	\caption{Generic receive signal flow diagram enabling post-experiment modification of forward error correction, symbol mapping, pre-coding, and constellation evaluation.}
\end{figure}

Shown in Figure 2 is a notional signal flow diagram for the receive processing chain that can make use of the additive dither and binary XOR operations for post-experiment testing.  As shown, the received signal could be received from an array of $K$ elements $\boldsymbol{r}(t)=[r_1 (t),\cdots,r_K (t)]^T$, bandshifted and sampled, (or, as is typically the case for acoustic communication systems, sampled at carrier, and digitally band-shifted and decimated) to yield the vector of received signals at complex baseband, $\boldsymbol{y}[n]=[y_1 [n],\cdots,y_K [n]]^T$. This use of a complex additive dither can enable symbol modulation and constellation regression to compare various trade-offs between code rate and symbol modulation order. Such methods might also be useful for testing dithers for decorrelating sources, to aid in multi-user environments and networking, as well as in adding and testing methods appropriate for clandestine operation.

To facilitate equalization methodologies that incorporate Doppler tracking and resampling, the lowpass filter and ADC block of Figure 2 should operate at a sufficient rate to accommodate any bandwidth expansion that might occur due to Doppler. If it is anticipated that substantial Doppler may be present and post-experiment algorithm development would benefit from ground-truth Doppler measurements, it may be feasible to transmit side-carriers just outside the signal band to enable sample-by-sample Doppler measurement, which could be calculated on a per-receiver, or per-beam basis for array receivers.

If other physical-layer modulation methods are employed, such as OFDM, a similar strategy to the above for QAM (single-carrier) modulation can be employed, as shown in Fig~\ref{fig:OFDM}. If the symbol-mapper in Figure 1 is interpreted as a standard symbol mapper, followed by serial-to-parallel converter, inverse DFT, parallel to serial conversion and cyclic prefix extension (or zero-padding), the same methods for both constellation and FEC experimentation can still be employed. In addition, with OFDM applications, by proper design of the FEC, either parity symbols, or known symbols can be inserted in specific OFDM carriers, enabling a host of channel estimation and Doppler correction strategies to be employed, again, without the need for re-transmission of the experimental data. Similar arrangements can be made by replacing the serial to parallel conversion in Fig.~\ref{fig:OFDM} with a “resource block mapper” that judiciously maps symbols to OFDM subcarriers. This block could either allocate parity or training data to specific carriers, or avoid certain carriers altogether for purposes of spectrum allocation. Such methods are common in both wireline (e.g. xDSL) and wireless (e.g., LTE) OFDM approaches.

\begin{figure}[h]
	\centering
	\includegraphics[scale=.3]{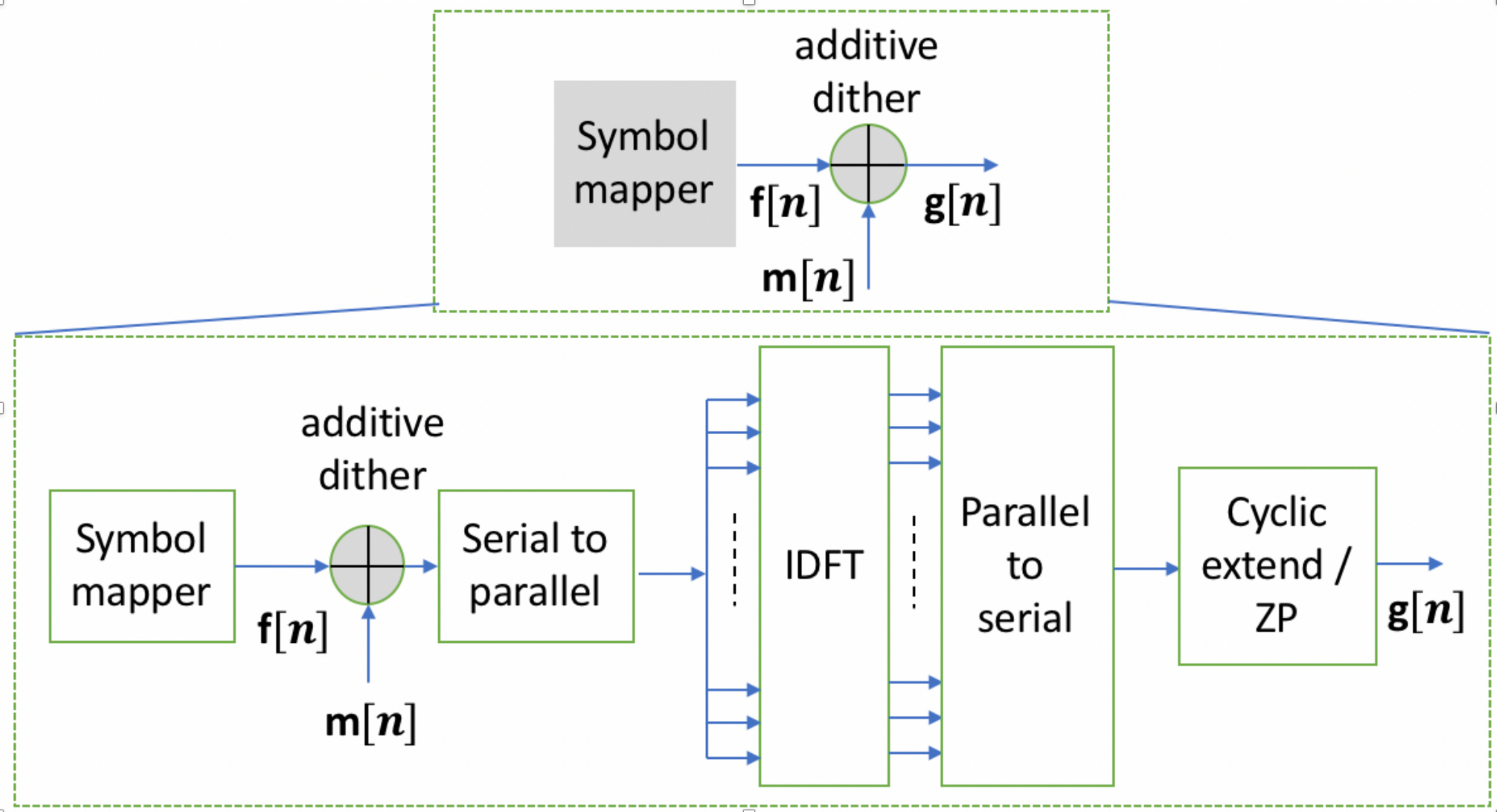}
	\caption{OFDM symbol mapper and additive dither.\label{fig:OFDM}}
\end{figure}

\subsection{Dither Sequence Design}
\begin{figure}[h]
	\centering
	\includegraphics[scale=0.5]{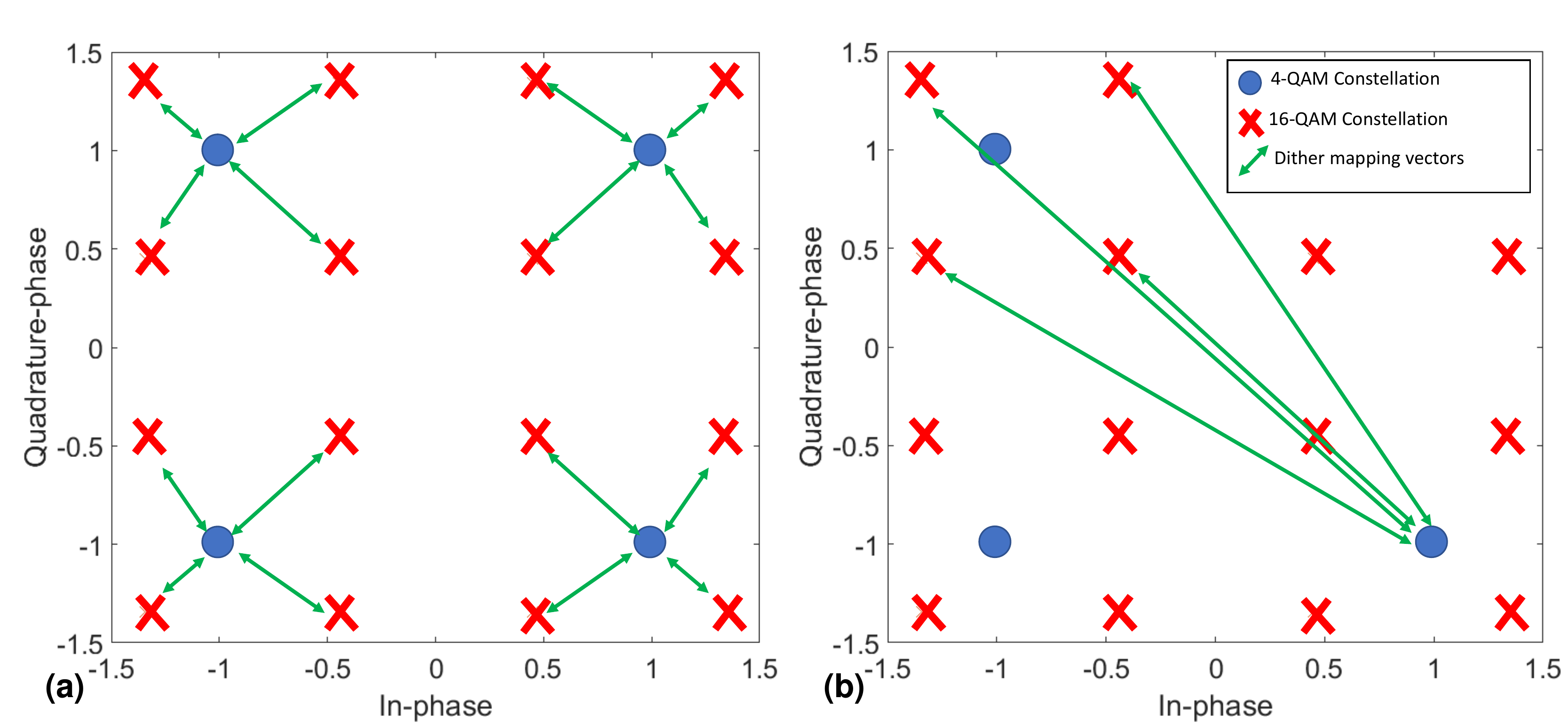}
	\caption{Different symbol mappings that correspond to different dither sequences between QPSK and 16QAM symbols with (a) minimum magnitude and (b) maximum magnitude for the symbol shown.\label{fig:dither}}
\end{figure}

In the previous section, two different elements were inserted into the signal path to enable the evaluation of communication schemes with different parameters. It is straightforward to see that the modification of the FEC through the use of an XOR with a whitening sequence will not change the result of experiment, since this is simply a re-interpretation of the meaning of the bits used to construct the channel symbols, and the performance evaluation and analysis of the system merely involves exchanging one encoding scheme for another. The only potential difficulties arise from codes of different rates leading to transmission lengths that differ from the original payload data transmission, but this can easily be handled through defining a different data sequence $b[n]$.

However, in order to test signals with different transmit signal constellations, a number of questions might arise. 
First, to mimic the desired experimental result, the inserted dither sequence at the transmitter side must be removed at the receiver. For this to be accomplished though a simple subtractive dither \emph{before equalization}, there would need to be no intersymbol interference (ISI) in the channel. For use with channels that do exhibit ISI, the dither must be \emph{incorporated into the equalizer}, as shown in Fig. 2, such that the equalizer, for example, forms an estimate of $\hat{\bf{g}}[n]$, from which $\hat{\bf{f}}[n]$ can be obtained through subtractive dither. There are a number of other ways in which the equalizer can incorporate the dither ${\bf m}[n]$ directly, though an implicit assumption in this process is that the underlying channel, from transmitter to receiver can be well-modeled as a linear, possibly time-varying, operation, such that superposition of the channel effects due to the immutable transmitted sequence ${\bf{g}}[n]$ can be algorithmically altered to mimic the transmission of the desired sequence ${\bf f}[n]$ through superposition (actually, only additivity is required). While the XOR processing described above for FEC experimentation required no such linearity assumption, substantial deviations from linearity would affect  conclusions formed based on an additive dither anaylsis.  Furthermore, while the choice of dither sequence ${\bf{m}}[n]$ is uniquely determined given ${\bf{f}}[n]$ and ${\bf{g}}[n]$, it is not unique given only the transmitted symbols ${\bf{g}}[n]$ and a desired alternate constellation in which the sequence ${\bf{f}}[n]$ must lie. The subsequent choice of  mapping used to generate the symbols in the new constellation from the sequence $p[n]$ may give rise to different dither sequences and have implications into transmit signal constraints and performance conclusions, as we will explore.

Let us assume that a set of data transmissions were conducted for waveforms generated from QPSK transmit signal constellations and the goal is to predict the behavior of this communication system had 16QAM symbols been transmitted instead, without undertaking another field experiment. In this case, any choice of symbol vectors comprising 16QAM symbols could be considered as having been virtually transmitted through proper definition of the dither sequence, ${\bf m}[n]$ as described above. In Figure \ref{fig:dither}, two different mappings between 16QAM and QPSK constellations are shown, e.g., when generating virtual 16QAM symbols, each transmitted QPSK symbol can be mapped to one of its nearest four neighbors in the 16QAM constellation, or to the furthest four locations, depending on the particular mapping that was used to create the 16QAM sequence. Intuitively, the nearest neighbor mapping in Figure \ref{fig:dither} (a) seems preferable over the choice described in (b), because the magnitude of the dither vector is smaller; and might have a smaller effect on any deviations from the expected results. In this section, this statement will be made more precise. Note that this nearest neighbor mapping can be obtained by selecting the four bits that comprise the 16QAM symbols such that the two bits that govern the quadrant are the same as the two bits that comprise the transmitted QPSK symbol, and the remaining two bits are taken from either the latter half of the data sequence (assuming the 16QAM sequence is half the length of the QPSK sequence), or from data not previously transmitted (assuming the 16QAM sequence is of the same length as the QPSK sequence and hence has twice the number of code bits $p[n]$).

\subsection{Problem Formulation}

\begin{figure}[h]
	\centering
	\includegraphics[scale=0.6]{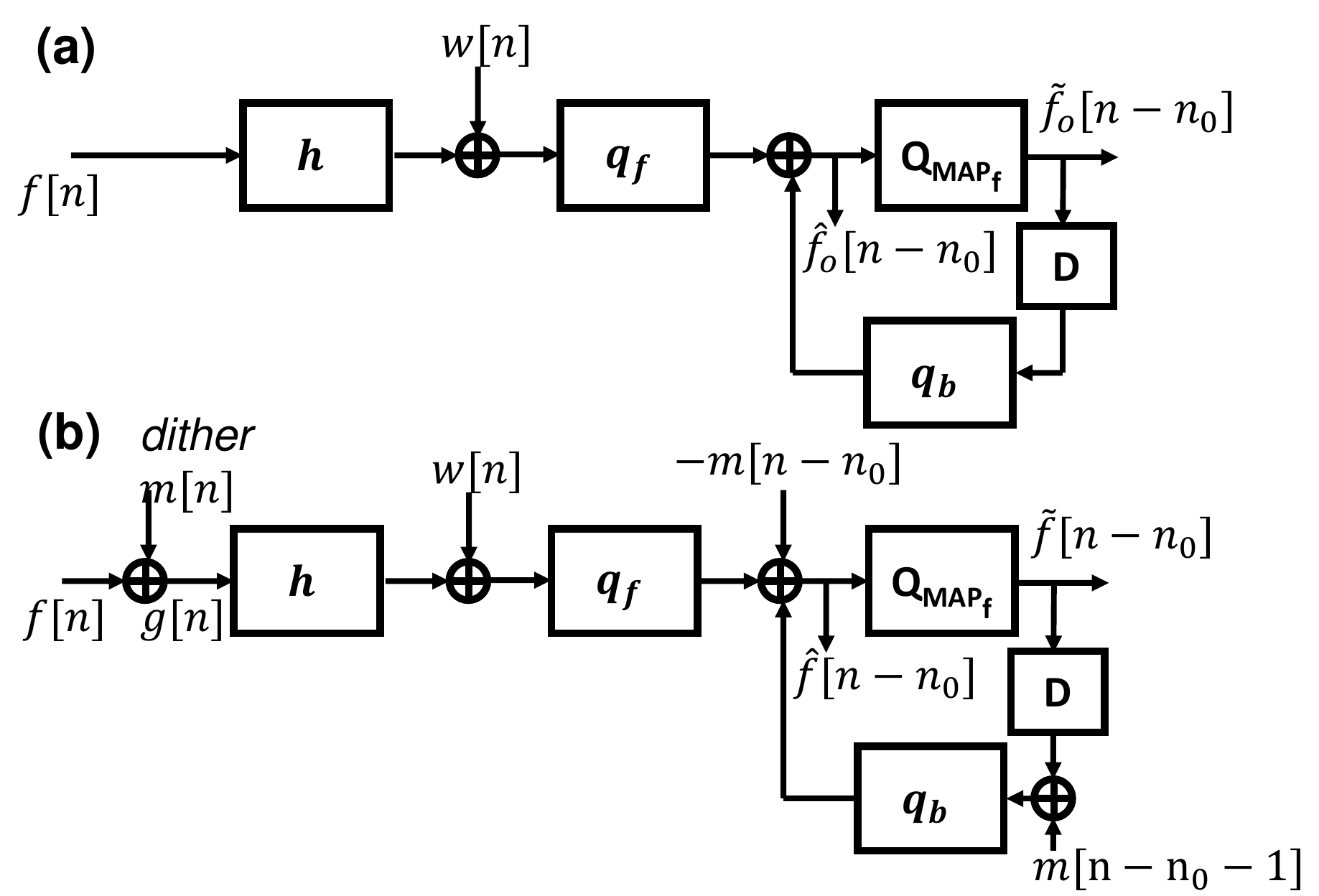}
	\caption{(a) Baseband schematic for point-to-point communication of the symbol sequence $f[n]$ with a decision feedback equalizing receiver (b) post-experimental modification of an experiment that transmitted $g[n]$ and virtual dither sequence insertion to emulate the transmission of $f[n]$.\label{fig:baseband}}
\end{figure}

A single-carrier point-to-point communication link is schematically shown with a decision feedback equalizer (DFE) receiver in Figure \ref{fig:baseband}. We assume a complex baseband system model with finite delay and Doppler spreads for the channel, resulting the (possibly multi-channel) waveform ${\bf y}[n]$ given as follows:
 
\begin{equation}\label{eq:linearchan}
    {\bf y}[n] = \sum_{l=0}^{L-1} {\bf h}_{n}[l] {\bf f}[n-l] + {\bf w}[n],
\end{equation}
  
where the length of time-varying FIR impulse response ${\bf h}_n[l]$ at time instant $n$ is limited by the maximum delay spread $L$. Additive noise ${\bf w}[n]$ need not be Gaussian, but is assumed i.i.d.  The time-varying impulse response is assumed to capture the end-to-end signal paths from source transducers to receive transducers, including any amplifiers, signal conditioning, and ocean-acoustic channel effects.

On the receiver side, the resulting intersymbol interference is equalized via a DFE structure with feedforward filters with $L_f$ taps and a feedback filter with $L_b$ taps with decision delay $n_0$. Because of the existence of the slicer during the equalization process, the DFE is essentially a nonlinear element (and therefore violates the assumption that superposition holds for end-to-end processing). However, in this work, we assume the system to be linear, which amounts to assuming that the DFE processes correct symbol decisions, as is common in DFE analysis \cite{proakis}. This can be seen by replacing the slicer with another path from the source directly to the equalizer with only the feedback filter and a delay in between, making the entire process once again linear. This is a reasonable assumption for operation at high signal-to-noise ratio (SNR), for which $\tilde{f}_o[n] \approx f[n]$. This also holds when the DFE is operated in training mode, when the slicer is replaced directly by known training data. Finally, the equalized symbol $\hat{f}_o[n]$ is given by

\begin{equation}\label{eq:fohat_mse}
\begin{aligned}
\hat{f}_o[n-n_0] &= q_f^{*}y_{n}^{n-L_f+1}-q_b^{*} \tilde{f_o}_{n-n_0-1}^{n-n_0-L_b}  \\
& \approx q_f^{*} \boldsymbol{H} f_{n}^{n-L_f-L+2}-q_b^{*} f_{n-n_0-1}^{n-n_0-L_b} + q_f^{*}w_{n}^{n-L_f+1},
\end{aligned}
\end{equation}

where $q_f$ and $q_b$ are $L_f \times 1$, $L_b \times 1$ vectors each representing feedforward and feedback filter taps, in the scalar case. A vector written, $x_{n_1}^{n_2}$,  is defined as $[x[n_1], x[n_1-1], ..., x[n_2]]^T $, and the superscript * denotes the Hermitian operator. Finally, symbol decisions, $\tilde{f}_o[n]$, are made by slicing (quantizing) the equalizer output $\hat{f}_o[n]$ to the nearest member of the transmitted symbol constellation. The matrix $\boldsymbol{H}$ refers $L_f \times (L_f+L-1)$ matrix for the time-varying channel model in \eqref{eq:linearchan}, i.e.,
 
\begin{equation}\label{eq:channel_matrix}
    \boldsymbol{H} = 
    \begin{bmatrix}
    h_n[0] & h_{n-1}[1] & h_{n-2}[2] & \dots  & 0 \\
    0 & h_{n-1}[0] & h_{n-2}[1] & \dots  & 0 \\
    \vdots & \vdots & \vdots & \ddots & \vdots \\
    0 & 0 & 0 & \dots  & h_{n-L_f-L+2}[L-1]
\end{bmatrix}.
\end{equation}

We consider that a baseband symbol sequence $g[n]$ following the symbol constellation $\bold{MAP_{g}}$, was originally sent via the system. Here, the sequence $g[n]$ is assumed to be white, i.e., statistically uncorrelated. After the experiment was completed, we consider a virtual transmission of a signal with the different symbol mapping $\bold{MAP_{f}}$, i.e., $\bold{MAP_{f}}\neq \bold{MAP_{g}}$ . Recall that the sequence from the new mapping, $f[n]$, can be arbitrarily chosen making use of the FEC modification technique, and the resulting bit sequence mapped onto that used in the experiment through a proper XOR sequence. Figure \ref{fig:baseband} (a) illustrates a virtual experiment, where the sequence $f[n]$ was desired as the transmitted sequence rather than the sequence $g[n]$ used at the time of the experiment. In our post-experimental scheme illustrated in Figure \ref{fig:baseband} (b), the resulting output $\hat{f}[n]$ is an approximation of the desired output $\hat{f}_o[n]$ in (a), i.e., we seek a dither sequence $\bar{m}$ given by
 
\begin{equation} \label{eq:optm}
    \bar{m} = \argmin_{m \in \bold{M}} \textrm{Dis}(\hat{f}_o, \hat{f}|g, m), 
\end{equation}
  
where $\textrm{Dis}(\hat{f}_o, \hat{f}|g, m)$ denotes a distance measure of interest between $\hat{f}_o$ and $\hat{f}$ for a given transmitted sequence $g$ and the choice of dither $m$. $\bold{M}$ is a set of white dither sequences such that $f=g+m$ for $f \in \bold{MAP_{f}}$. We next discuss two potential means for making this selection: the mean-squared deviation (MSD) and the Kullback-Leibler (KL) divergence.

\subsection{Minimum Mean Squared Deviation Criterion}

The output from post-experimental analysis $\hat{f}[n]$ is given by:
 
\begin{equation}\label{eq:fhat}
\begin{aligned}
\hat{f}[n-n_0]  &\approx q_f^{*} \boldsymbol{H} g_{n}^{n-L_f-L+2}-q_b^{*} g_{n-n_0-1}^{n-n_0-L_b} + q_f^{*}w_{n}^{n-L_f+1}-m[n-n_0] \\
& = q_f^{*} \boldsymbol{H} f_{n}^{n-L_f-L+2}-q_b^{*} {f}_{n-n_0-1}^{n-n_0-L_b} + q_f^{*}w_{n}^{n-L_f+1}+(q_f^{*} \boldsymbol{H} {m}_{n}^{n-L_f-L+2}-q_b^{*} {m}_{n-n_0-1}^{n-n_0-L_b}-m[n-n_0]) \\
&= \hat{f}_o[n-n_0]+(f^{*} \boldsymbol{H} {m}_{n}^{n-L_f-L+2}-b^{*} {m}_{n-n_0-1}^{n-n_0-L_b}-m[n-n_0]).
\end{aligned}
\end{equation}
  
The last equality comes from \eqref{eq:fohat_mse}. 

The strongest condition we can impose on the problem described in \eqref{eq:optm} is approximating $\hat{f}[n]$ with $\hat{f}_o[n]$, element by element. In this case, \eqref{eq:optm} can be rewritten as a  minimum mean squared error relation given by 
 
\begin{equation}
    \bar{m} = \argmin_{m \in \bold{M}} E(|\hat{f}[n]-\hat{f}_o[n]|^2).
\end{equation}
  
The mean squared deviation $E(|\hat{f}[n]-\hat{f}_o[n]|^2)$ can be readily computed from \eqref{eq:fhat} as
 
\begin{equation}\label{eq:ditheroutout}
    E(|\hat{f}[n]-\hat{f}_o[n]|^2) = E(|m[n]|^2)(q_f^{*}\boldsymbol{H}\boldsymbol{H}^*q_f+q_b^{*}q_b+1-2 \Re\{(q_f^{*}\boldsymbol{H}_1^{L_b-n_0-1}{q_b}_{n_0+2}^{L_b}-q_f^{*}\boldsymbol{H}_{n-n_0})\}),
\end{equation}
  
where the whiteness of $m[n]$ was used in the derivation and $\Re\{x\}$ denotes the real part of $x$. Here, $\boldsymbol{H}_{n-n_0}$ denotes the ${(n-n_0)}^\textrm{th}$ column of matrix $\boldsymbol{H}$ and $\boldsymbol{H}_a^b$ refers a submatrix composed of $a^\textrm{th}$ to $b^\textrm{th}$ columns of $\boldsymbol{H}$. It can be seen from \eqref{eq:ditheroutout} that the output deviation is proportional to the magnitude of the dither sequence, as suggested previously in this section. In other words, sufficient and necessary conditions for the dither $m[n]$ to minimize the mean squared deviation can be summarized in the following theorem; \begin{theorem}
Given a linear channel model as in \eqref{eq:linearchan}, a desired constellation for transmission of the sequence $f[n]$ and another  transmitted symbol sequence $g[n]$, let $\tilde{f}_o[n]$ be the DFE receiver output when $f[n]$ is transmitted directly, and let $\tilde{f}[n]$ be the DFE receiver output when dither $\bar{m}[n]=g[n]-f[n]$ is used to emulate transmission of $f[n]$. Then the optimal dither minimizing the mean squared deviation between the $\tilde{f}_o[n]$ and $\tilde{f}[n]$ is the one with the minimal $l_2$ norm and the resulting mean squared deviation (MSD) is given by \eqref{eq:ditheroutout}.
\end{theorem}
\begin{proof}
The derivation of \eqref{eq:ditheroutout} yields the output deviation as a function of dither sequence $m[n]$. It can be shown that that the second term on the right hand side of \eqref{eq:ditheroutout} is positive and hence the MSD is directly proportional to the $l_2$ norm of $\bar{m}[n]$, yielding the result.
\end{proof}
The equalizer is designed to efficiently mitigate the mean-square effects of ISI and as it does, the multiplicative term on the right side of \eqref{eq:ditheroutout} proportionally decreases (and becomes zero as ISI is cancelled), reducing the effect of $E(|m[n]|^2)$. The dither generation described in Figure \ref{fig:dither} (a) satisfies the condition in Thm. 1, hence, among all choices, guarantees our post-experimental evaluation will be closest to the desired experimental results, in an element-wise sense. This can be achieved whenever the constellation is mapped from a $2^k$-QAM to a $2^{mk}$-QAM constellation, or from a $2^{mk}$-QAM constellation to a $2^k$-QAM constellation. In the former case, $k$ bits from the binary sequence $p[n]$ are mapped into the symbols $f[n]$ and $mk$ bits are needed for the symbols $g[n]$. Assuming the sequence $f[n]$ is of length $mN$ and letting $g[n]=A(2^{m-1}f[n]+2^{m-2}f[n+N]+\cdots+f[n+(m-1)N])$ yields the sequence $g[n]$ in the same (sub-)quadrant as $f[n]$, where $A$ is a suitably chosen energy normalization. In the latter case, $mk$ bits from the binary sequence $p[n]$ are mapped into the symbols $f[n]$ and only $k$ bits are needed for $g[n]$. Truncating the representation to include only the first $k$ bits of $p[n]$ mapped into $f[n]$, yields the sequence $g[n]$ in the same (sub-)quadrant as $f[n]$. An example of this procedure for $2^2$-QAM (QPSK) and $2^4$-QAM is shown in Figure \ref{fig:constellationmap}.

\begin{figure}[h]
	\centering
	\includegraphics[scale=0.8]{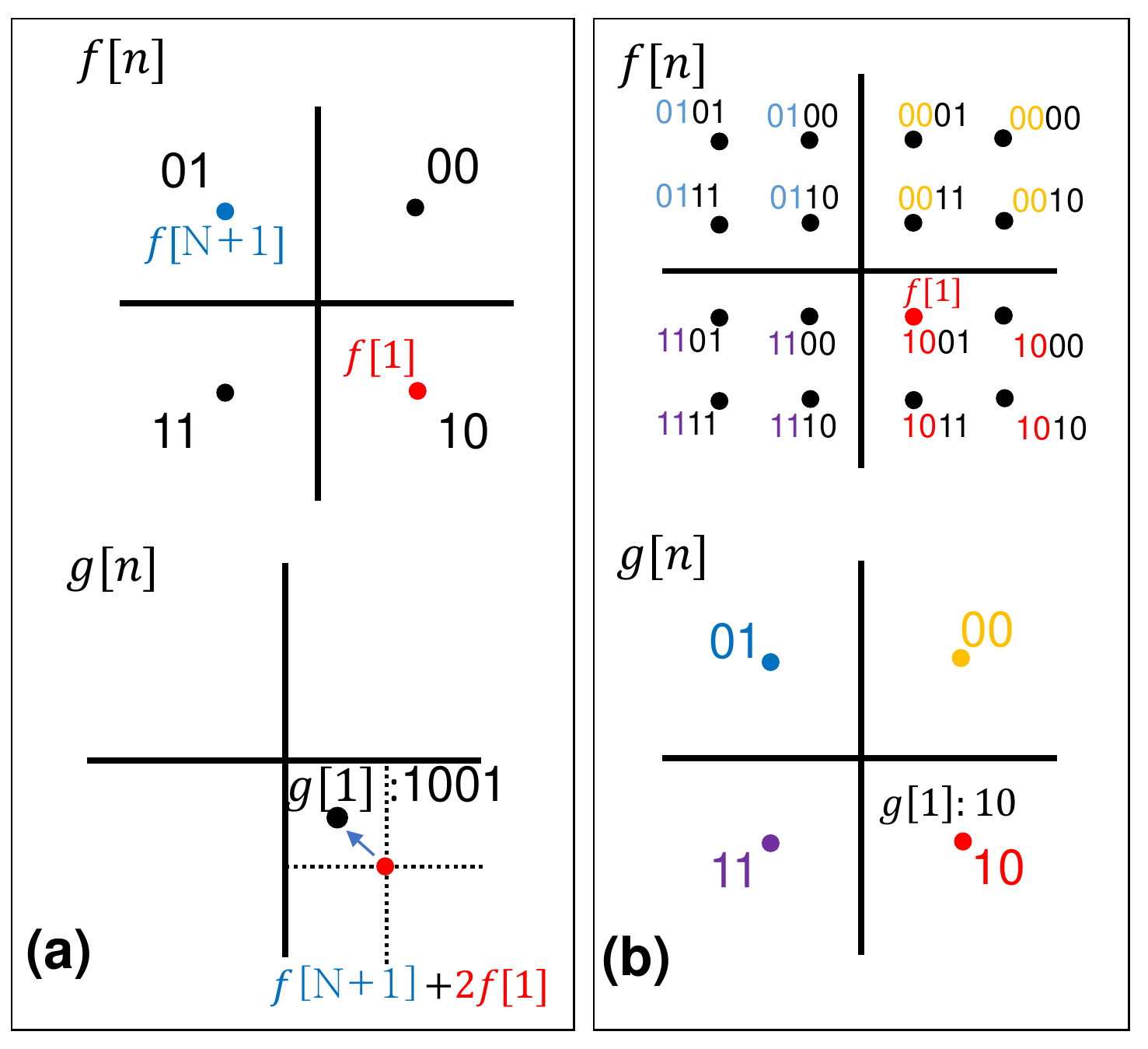}
	\caption{(a) Originally transmitted QPSK constellation used to transmit $f[n]$ and dither-generated 16QAM constellation for $g[n]$ (b) Originally transmitted 16QAM constellantion used to transmit $f[n]$ and dither-generated QPSK constellation for $g[n]$.\label{fig:constellationmap}}
\end{figure}

\subsection{Minimum KL Divergence Criterion}

Rather than enforcing that our experiment emulation process be faithful on a symbol-by-symbol basis, we often are only concerned with the system performance as measured statistically over the entire transmission. Hence, we are  frequently more concerned that the our experiment emulation faithfully reproduce the distribution of outcomes, rather than in an element-wise sense. For example, we are often interested in the bit error rate (BER) computed empirically over long sequences and hope they provide an accurate representation of that which would have occurred under experimental conditions. As a result, for each symbol $k\in\bold{MAP_f}$, it is often sufficient to require equality in distribution, i.e., $p(\hat{f}[n]|f[n]=k) = p(\hat{f_o}[\ell]|f[\ell]=k), n\neq \ell$, where $p(\cdot|\cdot)$ denotes a conditional density. 

Now, the measure in \eqref{eq:optm} is interpreted as a statistical distance between two empirical distributions, $p(\hat{f}[n]|f[n]=k)$ and $p(\hat{f_o}[\ell]|f[\ell]=k)$. For instance, KL divergence can be used to redefine the optimization problem \eqref{eq:optm} as follows:
 
\begin{equation} \label{eq:mKL}
     \bar{m}_{KL} = \argmin_{m \in \bold{M}} \sum_{k\in\bold{MAP_f}} p(f[n]=k) D_\text{KL}(p(\hat{f}[n]|f[n]=k)\parallel p(\hat{f_o}[\ell]|f[\ell]=k)),
\end{equation}
  
where $D_\text{KL}(\cdot\parallel \cdot)$ denotes the KL divergence. The right side of \eqref{eq:mKL} is the average KL divergence between the desired output $\hat{f}_o[n]$ and post-experimental output $\hat{f}[n]$ conditioned on the virtual transmitted symbol $f[n]$, averaged over the symbols in the constellation $\bold{MAP_f}$.  Note that the conditioned distribution is used in place of the marginal distributions $p(\hat{f}[n])$ and $p(\hat{f}_o[n])$, to preserve the desired transmitted signal; for example, the two marginal distributions might be equivalent under transmit symbol permutation. However, permutation of the symbols loses the transmitted information content.

Approximating the distributions as conditionally Gaussian \cite{tuchler2011turbo}, the KL divergence can be readily computed. When conditioned on $f[n-n_0]=k$, i.e.,  virtual symbol $k$ was transmitted at time instant $n-n_0$, then
 
\begin{equation} \label{eq:fohatcondition}
\hat{f}_o[n-n_0]\biggr\rvert_{f[n-n_0]=k} = q_f^{*}\boldsymbol{H}_{n_0+1} k + q_f^{*} \boldsymbol{H}^{n-n_0} f^{n-n_0}-q_b^{*} f_{n-n_0-1}^{n-n_0-L_b} + q_f^{*}w_{n}^{n-L_f+1},
\end{equation}
  
where $\boldsymbol{H}^{n-n_0}$ is a submatrix composed of all columns of $\boldsymbol{H}$ except for the ${(n-n_0)}^\textrm{th}$ column and, similarly, $f^{n-n_0}$ refers to a subvector whose ${(n-n_0)}^\textrm{th}$ element is removed from $f$. Since $f[n]$ is assumed to be white, hence an independently distributed random process, and $w$ is an i.i.d. noise vector, the empirical distribution of $\hat{f}_o[n]$ is well-approximated as Gaussian, as the channel length $L$ and equalization filter dimension $L_f$, $L_b$ grow, as the Berry-Esseen theorem \cite{esseen1945fourier} suggests, and as illustrated in the following lemma. \begin{lemma}
Let $\boldsymbol{v} = [v_1, v_2, ... , v_n]$ be an $n\times1$ row vector with non-zero entries, and which satisfies $\sum_{i=1}^{n} v_i^2 = c < \infty$ and $\max_{j,k} |v_j/v_k| = d < \infty$. Then, for i.i.d. random variables $f_1, f_2, f_3, ... , f_n$ with $E(f_i) = 0$, and $E(|f_i|^2) = 1$, $\sum_{i=1}^n v_i f_i$ converges to the Gaussian random variable $\mathcal{N}(0,c)$ in distribution as $n \rightarrow \infty$.
\end{lemma}
\begin{proof}
Denote $F_n$ and $\Phi$ as the cumulative distribution functions of $S_n = \sum_{i=1}^n v_i f_i/\sqrt{c}$ and $\mathcal{N}(0,1)$, respectively. Then, by the Berry-Esseen theorem, there exists a $C_0$ that satisfies for all $n$\\
 
\begin{equation}
\sup_{x\in \Omega}\left|F_n(x) - \Phi(x)\right| \le C_0\cdot\psi_0
\end{equation}
  
holds, where 
 
\begin{equation}
\psi_0=c^{-3/2}\cdot\sum\limits_{i=1}^n|v_i|^3 \leq c^{-3/2} \cdot \max_{j} |v_j| \sum_{j=1}^n|v_i|^2 \leq c^{-1/2} \cdot d \min_{j} |v_j| \leq \frac{d}{\sqrt{n}}.
\end{equation}
  
The first equality comes from the Berry-Esseen theorem, the second inequality follows from $\max_{j,k} |v_j/v_k| = d < \infty$, and $\sum_{i=1}^{n} v_i^2 = c < \infty$ and the last inequality follows from bounding $\min_j |v_j|$ using its square summability again. Therefore, the $\sum_{i=1}^n v_i f_i$ converges to a Gaussian random variable in distribution as $n \rightarrow \infty$.
\end{proof}
The conditions described in the lemma hold for the row vectors $q_f^* \boldsymbol{H}_{n_0+1}$ and $q_b^*$ in \eqref{eq:ditheroutout}, when the impulse response is square summable. Therefore, for channels with long delay spreads, typical in underwater acoustic communications, the Gaussian assumption on $\hat{f}_o$ would be reasonable for our analysis. It should be noted that conditioning on the event $f[n-n_0] = k$ plays a critical role in the Gaussian approximation. Without conditioning, the term $f[n-n_0]=k$ in \eqref{eq:fohatcondition} is a random variable which exhibits highly non-Gaussian structure in the received symbol constellation. Similarly $\hat{f}[n]$, i.e.,
 
\begin{equation} \label{eq:10}
\hat{f}[n-n_0]\biggr\rvert_{f[n-n_0]=k} = \hat{f}_o[n-n_0]+ q_f^{*}\boldsymbol{H}_{n_0+1} m[n-n_0] + q_f^{*} \boldsymbol{H}^{n-n_0} m^{n-n_0}-q_b^{*} m_{n-n_0-1}^{n-n_0-L_b}, 
\end{equation}
  
can also be approximated as a Gaussian process. 

With the Gaussian assumption, the target function in \eqref{eq:mKL} can be put into closed form as follows \cite{thomascover}:
 
\begin{equation} \label{eq:11}
       \bar{m} = \argmin_{m \in \bold{M}} \sum_{k\in\bold{MAP_f}} p(f[n]=k) \bigg( \log\left(\frac{\sigma_{2,k}}{\sigma_{1,k}}\right) + \frac{\sigma_{1,k}^2 + (\mu_{1,k}-\mu_{2,k})^2}{2\sigma_{2,k}^2} - \frac{1}{2} \bigg),
\end{equation}
  
where $\hat{f}[n-n_0]\biggr\rvert_{f[n-n_0]=k} \sim \mathcal{N}(\mu_{1,k}, \sigma_{1,k}^2)$, and $\hat{f}_o[n-n_0]\biggr\rvert_{f[n-n_0]=k} \sim \mathcal{N}(\mu_{2,k}, \sigma_{2,k}^2)$. The means $\mu_{1,k}$ and $\mu_{2,k}$ and the variances $\sigma_{1,k}^2$ and $\sigma_{2,k}^2$ of each distribution can be derived as
 
\begin{equation}\label{eq:12}
\mu_{1,k} = \mu_{2,k}+q_f^{*}\boldsymbol{H}_{n_0+1} E_{m,k},
\end{equation}

\begin{equation}\label{eq:13}
\mu_{2,k} = q_f^{*}\boldsymbol{H}_{n_0+1} k,
\end{equation}

\begin{equation} \label{eq:14}
\begin{aligned} 
\sigma_{1,k}^2 &= E(|w[n]|^2)q_f^*q_f+C \cdot E(|g[n]|^2)-|(q_f^* \boldsymbol{H}_{n_0+1}-1)E_{g,k}|^2 \\
&\quad -2\Re\{((q_f^* \boldsymbol{H}_{n_0+1}-1)E_{g,k}E_{m,k}^*)\}-|E_{m,k}|^2.
\end{aligned}
\end{equation}

\begin{equation} \label{eq:15}
\sigma_{2,k}^2 = E(|w[n]|^2)q_f^*q_f+C \cdot E(|f[n]|^2)-|(q_f^* \boldsymbol{H}_{n_0+1}-1)k|^2, 
\end{equation}
  
where $E_{m,k} = E(m[n]\rvert f[n]=k)$, and $E_{g,k} = E(g[n]\rvert f[n]=k)$ and the constant $C$ in \eqref{eq:14} and \eqref{eq:15} refers to the normalized mean-squared error interference which arises after the equalizer, independent of dither sequences, i.e., $C = (q_f^{*}\boldsymbol{H}\boldsymbol{H}^*q_f+q_b^{*}q_b+1-2 \Re\{(q_f^{*}\boldsymbol{H}_1^{L_b-n_0-1}{q_b}_{n_0+2}^{L_b}-q_f^{*}\boldsymbol{H}_{n-n_0})\})$. From these expressions, we can see that a substantial portion of the deviation arises from the equalizer bias $(q_f^* \boldsymbol{H}_{n_0+1}-1)$. The derivation of the above parameters can be found in Appendix A. If we substitute \eqref{eq:12}-\eqref{eq:15} into \eqref{eq:11}, we obtain the following theorem:
\begin{theorem}
Assuming that the energy of an experimentally and a virtually transmitted symbol sequence are equal, i.e., $E(|g[n]|^2)= E(|f[n]|^2)$, then, if there exists a mapping $m\in \bold{M}$ that satisfies $E_{m,k} = E(m[n]\rvert f[n]=k) =0$ for $\forall k \in \bold{MAP_f}$, then $m$ minimizes \eqref{eq:11}, and well-approximates the optimal dither in a KL divergence perspective, $\bar{m}_{KL}$ in \eqref{eq:mKL} as the channel length $L$ and equalization filter dimension $L_f$, $L_b$ grows under the conditions illustrated in Lemma 1.
\end{theorem}
\begin{proof}
As $L$, $L_f$, and $L_b$ grow under the conditions held in lemma 1, $p(\hat{f}[n]|f[n]=k)$ and $p(\hat{f}_o[n]|f[n]=k)$ approach the Gaussian distributions $p_f \sim \mathcal{N}(\mu_{1,k}, \sigma_{1,k}^2)$ and $p_{fo} \sim \mathcal{N}(\mu_{2,k}, \sigma_{2,k}^2)$, respectively.  As the dither $\bar{m}$ minimizes $D(p_f||p_{fo})$, it is enough to show that 
 
\begin{equation} \label{eq:16}
\begin{aligned} 
        \lim_{L, L_f, L_b \rightarrow \infty} D(p(\hat{f}|f) || p(\hat{f}_o|f)) 
    &= D(p_f||p_{fo}).
\end{aligned}
\end{equation}

In general, convergence in distributions does not guarantee convergence of the KL divergence between such distributions. However, \eqref{eq:16} can be ensured, when, conditioned on $f[n]=k$, $\hat{f}[n]$ and $\hat{f}_o[n]$ have the same mean, bounded fourth moments, and their distributions are bounded and continuously differentiable, which can be readily satisfied from the assumptions of the Theorem and 
Theorem 2 of \cite{moulin2014kullback}.  
A sketch of the proof in \cite{moulin2014kullback} is as follows. To prove \eqref{eq:16} we need to move the limit inside the integral for the KL divergence, i.e., 
 
\begin{equation} \label{eq:lebesgueconvergence}
\begin{aligned}
    \lim_{L, L_f, L_b \rightarrow \infty} \int_{x \in \Omega} p(\hat{f}|f)\log \frac{p(\hat{f}|f)}{p(\hat{f}_o|f)} &= \int_{x \in \Omega} \lim_{L, L_f, L_b \rightarrow \infty} p(\hat{f}|f)\log \frac{p(\hat{f}|f)}{p(\hat{f}_o|f)} \\
    &= \int_{x \in \Omega} p_f \log \frac{p_f}{p_{f_o}}.
\end{aligned}
\end{equation}
  
If $p(\hat{f}|f)\log\frac{p(\hat{f}|f)}{p(\hat{f}_o|f)}$ is bounded  on $\Omega$, then \eqref{eq:lebesgueconvergence} is guaranteed by the dominated convergence theorem. Dividing the integration domain $\Omega$ into two regions, $\Omega_1 \subset \Omega$, which is an arbitrarily large compact set, and $\Omega_2 = \Omega \setminus \Omega_1$,
 
\begin{equation} \label{eq:lebesgueconvergence2}
\begin{aligned}
    \lim_{L, L_f, L_b \rightarrow \infty} \int_{x \in \Omega} p(\hat{f}|f)\log \frac{p(\hat{f}|f)}{p(\hat{f}_o|f)} &= \lim_{L, L_f, L_b \rightarrow \infty} \Bigg\{
 \int_{x \in \Omega_1}  p(\hat{f}|f)\log \frac{p(\hat{f}|f)}{p(\hat{f}_o|f)}+\int_{x \in \Omega_2}  p(\hat{f}|f)\log \frac{p(\hat{f}|f)}{p(\hat{f}_o|f)} \Bigg\} \\
    &= \int_{x \in \Omega_1} p_f \log \frac{p_f}{p_{f_o}}+ \lim_{L, L_f, L_b \rightarrow \infty} \int_{x \in \Omega_2}  p(\hat{f}|f)\log \frac{p(\hat{f}|f)}{p(\hat{f}_o|f)},
\end{aligned}
\end{equation}
  
where the second equality comes from the assumption that $\Omega_1$ is compact, and hence the integrand is absolutely bounded over the domain $\Omega_1$. 
Also, It is possible to make the integration over $\Omega_2$ arbitrarily small as $L, L_f, L_b \rightarrow \infty$ by using exponential bounds of $p(\hat{f}|f)$ and $p(\hat{f}_o|f)$, making the second term in the right side of \eqref{eq:lebesgueconvergence2} converging zero. 
\end{proof}
This means that to minimize the deviation of the post-experimental analysis from the real experiment in a distribution sense, the best strategy is to design dither sequences that are well-balanced, while the dither sequence magnitude was important in an element-wise sense. As we can see from \eqref{eq:12} and \eqref{eq:13}, the well-balanced property plays a role in mitigating the bias of the equalizer output which produces the most substantial deviation, when $E(|g[n]|^2)= E(|f[n]|^2)$.  Also, note that a mapping rule satisfying $E_{m,k}=0$ may not always exist, e.g., when $\bold{MAP_g}$ is 16QAM, and $\bold{MAP_f}$ is QPSK.

This approach to post-experimental performance evaluation has advantages over direct simulation using an estimated channel based on measurements. Note that in this approach, since there is no channel estimation, channel estimation errors do not come into play.  However, practical implementation of this approach could yield larger deviations than expected from the analysis above for the following reasons: 1) at high-SNR, residual (uncompensated) ISI dominates over the effects of additive noise, 2) nonlinearities in the receiver, including that of the slicer (and slicer errors) in a decision feedback equalizer were not considered in the analysis, and 3) signal path nonlinearities, including the transmit amplifiers,  transducers and other components, violate superposition and deviate from the analysis above.

\subsection{Distortion Due to Amplifier Non-linearity}

Amplifier non-linearity introduces a natural deformation to the otherwise linear communication pipeline. Consider the following model for amplifier non-linearity imposed on the modulating symbols $g[n]$ with the understanding that the effective range of the constellation, $||g||_{\infty}\triangleq \max |g|$, is known. Here, the deformed signal constellation, denoted by  $\hat{g}[n]$ takes the following form:
 
\begin{equation}
    \tilde{f}[n] =  ||f||_{\infty}\tanh{\left(\frac{f[n]}{\alpha||f||_{\infty}}\right)},
\end{equation}
  
where, $\alpha>1$ denotes the effective range to occupied effective range ratio for the amplifier. Symbols of higher energy are are subject to more severe distortion and symbols of lower energy remain relatively undistorted. Such a distortion, changes effective decision making regions for receiver. Furthermore, the error analysis for the linear pipeline needs to be extended. Consider that:

\begin{equation}
    \tilde{f}[n] = ||g+m||_{\infty}\tanh{\left(\frac{g[n]+m[n]}{\alpha||g+m||_{\infty}}\right)} = (g[n]+m[n]) - \frac{(g[n]+m[n])^3}{\alpha^2||g+m||^2_{\infty}} + \textsc{H.O.T.}
\end{equation}
  
Here the last step follows from the Taylor series expansion of the non-linearity and higher order terms (H.O.T.) continue on odd powers. Naturally, such a deformation of transmitted symbols changes the errors suffered from the modified communication pipeline. The mean-square error, up to second moments are modified as:

\begin{equation}
    E(|\tilde{f}[n]-\hat{f}_o[n]|^2) \approx E(|m[n]|^2)((q_f^{*}\boldsymbol{H}\boldsymbol{H}^*q_f+q_b^{*}q_b)(1-\frac{1}{\alpha^4})+1-2 \Re\{(q_f^{*}\boldsymbol{H}_1^{L_b-n_0-1}{q_b}_{n_0+2}^{L_b}-q_f^{*}\boldsymbol{H}_{n-n_0})\}),
\end{equation}
  
For amplifiers with more reliable effective ranges ($\alpha \gg 1$), the error expression converges to that of \eqref{eq:ditheroutout}.

In this section, we introduced field data reuse techniques that insert two additional elements into a standard signal transmission flow. First, by applying an XOR with an arbitrarily chosen binary sequence $s[n]$ to the transmitted bit sequence $p[n]$, we can validate the performance of virtual systems using practically any FEC scheme. Second, we extended this concept to incorporate different symbol mappings by inserting a dither addition stage during modulation. In Table 1, these signal preparation techniques are summarized, pointing out how to properly design them and under what conditions they can be applied to predict system performance. 

\begin{table}[H]
\centering
\captionsetup{justification=centering}
\caption{Summary of different signal preparation techniques for system evaluation with different coding and modulation schemes.}
\label{my-label}
\begin{tabular}{|m{2.8cm}|m{4.5cm}|m{4.5cm}|m{4.5cm}|}
\cline{1-4}
& \multicolumn{1}{c|}{XOR sequence} & \multicolumn{2}{c|}{Dither insertion} \\ \hline \hline
\multicolumn{1}{|l|}{Design}    & $s[n] = d_2[n] 
\text{''XNOR''} d[n]$   & $m[n] = g[n]-f_2[n]$ with minimal $l_2$ norm           & $m[n] = g[n]-f_2[n]$ with $E_{m,k}=0$            \\ \hline
\multicolumn{1}{|l|}{Assumptions} & $p[n]$ white & end-to-end linearity & 1. end-to-end linearity     $\,\,\,\,$   2. $E(|g[n]|^2) = E(|f[n]|^2)$ 3. Normal approximation (see Lemma 1)\\ \hline
\multicolumn{1}{|l|}{Purpose}  & Test different error-correction encoding and interleaving schemes  & Test different symbol mappings, and minimize \textbf{mean square deviation} between  virtual and actual experimental results     & Test different symbol mappings, and minimize \textbf{KL divergence} between virtual and actual experimental results          \\ \hline
\end{tabular}
\end{table}

\section{Numerical Examples}
%
To illustrate the proposed post-experimental evaluation scheme, data transmission through the structure illustrated in Figure 3 was first simulated via a synthetic channel. The baseband channel impulse response was modeled with a 5-tap filter, having a weights $ [1, 0.9472, 0.4586, 0.4315, 0.1497]^T$, and additive white Gaussian noise (AWGN) was assumed. This channel was simulated by switching modulation schemes between QPSK and 16QAM and varying SNR. When switching between modulation schemes, the same realization of AWGN was applied, since the data transmission was assumed conducted only once for a single symbol mapping; as illustrated earlier in Figure \ref{fig:dither}. Filters for the DFE were designed using three different criteria: zero-forcing (ZF), minimum mean squared error (MMSE), and the least-mean-square (LMS) algorithm. When designing the ZF-DFE and MMSE-DFE, perfect channel knowledge was assumed for illustration purposes.

Two situations which happen frequently in field experiments include collecting data for which the transmitted signal constellation is either too aggressive for the experimental conditions, and the data is unrecoverable, and collecting data under too conservative signal parameters that fail to challenge the underlying algorithms. In both of these situations, postmortem data analysis would have indicated a desire for additional transmissions with either a higher-order or a lower-order transmit signal constellation.  Use of the XOR dither sequence can enable other FEC strategies to be employed, which may enable recovery of data transmitted under constellations too aggressive for the data rates anticipated. However this may not be sufficient to overcome the environmental conditions. Similarly, FEC rate changes cannot be more aggressive than uncoded transmission, and hence fail to help with environments that proved too benign.  One approach that is often employed in this situation is available when a multi-channel array receiver is used. In this case, receiver strategies using only a subset, or a single receiver element may be employed. Similarly, recordings of environmental noise can be added to reduce the SNR of the environment.  In both of these situations, the proposed additive dither method provides accurate performance predictions without the need for collecting an exhaustive set of transmit symbol constellations and space-time transmissions to be captured during the experiment.

For this example, two different dither schemes were adopted for various equalizer structures, derived from the previous section. Here, $g[n]$ denotes the transmitted sequence in the experiment. Then the desired sequence $f[n]$ and dither $m[n]$ can be chosen as follows. When the constellation mapping for $g[n]$, $\bold{MAP_g}$, is 16QAM, and that for $f[n]$, $\bold{MAP_f}$ is assumed to be QPSK, the nearest QPSK point to $g[n]$ can be selected as $f[n]$ following the rule shown in Figure \ref{fig:dither} (a), and we call the dither sequence generated from this rule as $m_{1, 16\rightarrow 4}$. However, as illustrated in Figure \ref{fig:dither} (b), it is also possible to select the constellation point furthest from the original location, and the corresponding dither is referred as $m_{2, 16\rightarrow 4}$. In the other situation, where the mappings of $g[n]$ and $f[n]$ are reversed, if we allocate only a single point to each constellation location in $g[n]$, without changing the underlying rate of the system, then it would result that only 4 symbols from 16QAM are used, effectively maintaining QPSK. Therefore, we can either group two QPSK symbols into a single 16 QAM symbol, thus reducing the transmit sequence length in half, or, for illustration purposes, we can randomly choose among the 4 nearest 16QAM points to the QPSK symbols $g[n]$ with equal probability. The index of this selection corresponds to an additional 2 bits of data that could be conveyed in the transmission and this bit sequence can similarly be mapped to any desired transmit sequence of bits through the XOR dither described previously. In either case, we call the corresponding dither as $m_{1, 4\rightarrow 16}$. Finally, $m_{2, 4\rightarrow 16}$ refers the sequence which maps $g[n]$ among 4 furthest 16QAM points.

\subsection{Proof of Concept}
\subsubsection{Simulating QPSK transmission from a failed 16QAM transmission}
In the attempt to send a 16QAM signal through the synthetic channel using the MMSE-DFE structure, it was shown that 10-dB transmit SNR was not sufficient to reliably recover the data. As illustrated in Figure \ref{fig:10db} (a), the equalizer output shows a failure in restoring transmitted symbols. Through additive dither, this 16 QAM transmission was mapped to a QPSK transmission and subsequently processed using the MMSE DFE. The results are shown in Figure \ref{fig:10db} (b). Note that although the QPSK constellation that is constructed from the transmitted 16 QAM transmission shows a slight bias, the resulting equalized constellation closely mirrors that which would have been obtained through transmission of a QPSK constellation through the channel directly.

\begin{figure}[h]
	\centering
	\includegraphics[scale=0.8]{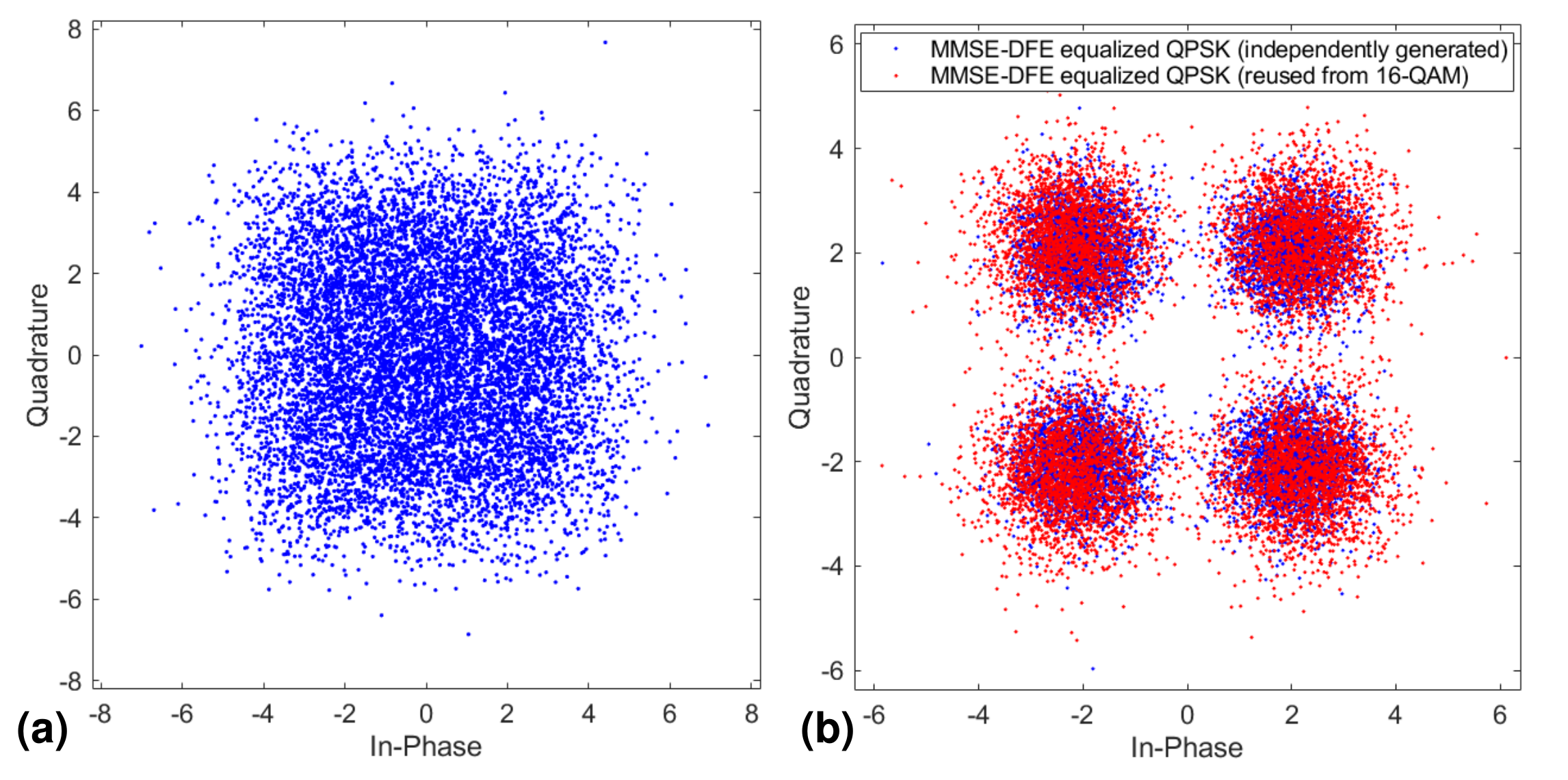}
	\caption{(a) The output constellation from 10 dB 16QAM, (b) the 10 dB QPSK output reproduced from (a) with the proposed method and comparison to the independently generated QPSK output without dither.\label{fig:10db}}
\end{figure}

\subsubsection{Simulating 16QAM transmission from a successful QPSK transmission}

\begin{figure}[h]
	\centering
	\includegraphics[scale=0.8]{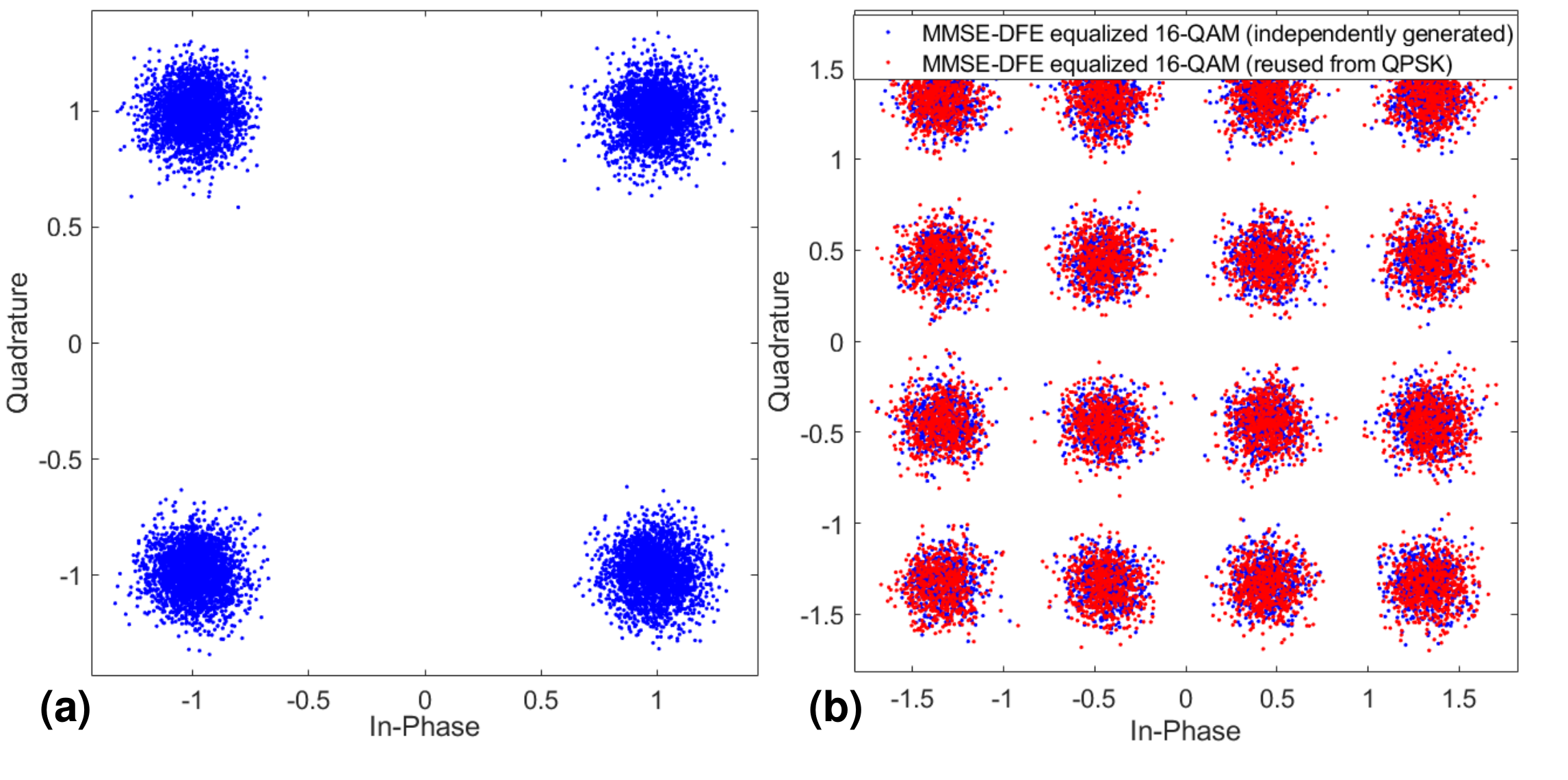}
	\caption{(a) The output constellation from 20-dB QPSK, (b) the 20-dB 16QAM output reproduced from (a) with the proposed method and its comparison to the independently generated 16QAM output without dither additions.\label{fig:20db}}
\end{figure}

In the attempt to send a QPSK signal through the synthetic channel using the MMSE-DFE structure, we see that 20 dB transmit SNR is overly conservative and it is clear that a higher-order constellation, such as 16 QAM could have been successful. As illustrated in Figure \ref{fig:20db} (a), the equalizer output is closely-grouped around the QPSK constellation. Through additive dither, this QPSK transmission was mapped to a 16 QAM transmission and subsequently processed using the MMSE DFE. The results are shown in Figure \ref{fig:20db} (b). Note that the resulting equalized constellation closely mirrors that which would have been obtained through transmission of a 16QAM constellation through the channel directly.

\subsection{BER versus SNR Prediction}

By sweeping SNR from 5 dB to 20 dB, bit error rates of the equalized sequences were obtained for both those reproduced via dither and independently generated reference transmissions. Both ZF-DFE and MMSE-DFE were used to equalize the received symbols. As can be seen from Figures \ref{fig:fig7} and \ref{fig:fig8}, the MMSE scheme outperforms the ZF-DFE in terms of mean squared error. However, the MMSE-DFE is essentially a biased estimator, i.e., output estimation errors have a non-zero offset. As discussed previously, this implies an increase in statistical distance as seen in equation of \eqref{eq:11}. Hence, the MMSE-DFE structure may be more prone to poor dither sequences than the ZF-DFE. This behavior is seen in Table 2, where empirically determined MSD and KL-divergence are shown between the output, $\hat{f}[n]$, and the reference output, $\hat{f}_o[n]$ for an example with 18dB  SNR. We  see not only that $D_{\text{KL}}$ is larger in MMSE structures than in ZF structures, but also, when a mapping scheme with larger $E(|m|^2)$ and $E_{m,k}$ was used, the distance increases more than that for the ZF-DFE. In addition, we observe larger MSD values in the MMSE-DFE than in the ZF-DFE. However, this phenomenon is not due to the bias. The multiplier in \eqref{eq:ditheroutout}, which determines MSD for a fixed dither magnitude, can be interpreted as optimal interference plus noise from the equalizer when zero noise is assumed, i.e., this is an exact condition for the ZF-DFE.

\begin{table}[H]
\centering
\captionsetup{justification=centering}
\caption{Empirically Computed Distances between the reproduced and the reference output for the transmit signal with 18dB SNR.}
\label{my-label}
\begin{tabular}{|P{3cm}|P{2cm}|P{2cm}|P{2cm}|P{2cm}|}
\hline
\multicolumn{1}{|c|}{\multirow{2}{*}{}} & \multicolumn{2}{l|}{\hspace{1.8cm}MSD} & \multicolumn{2}{l|}{ \hspace{1.8cm} $D_{\text{KL}}$} \\ \cline{2-5} 
\multicolumn{1}{|c|}{}                  & ZF-DFE    & MMSE-DFE           & ZF-DFE          &MMSE-DFE           \\ \hline \hline
$m_{1, 16\rightarrow 4}$(optimal)                 &1.62E-06  &4.74E-04 &3.84E-11    & 3.81E-06          \\ \hline
$m_{2, 16\rightarrow 4}$(poor)                 &1.57E-05  &4.41E-03   &5.11E-10   & 2.52E-04          \\ \hline
$m_{1, 4\rightarrow 16}$(optimal)                 &1.64E-06  &4.84E-04 & 4.42E-11   &  3.07E-06  \\ \hline
$m_{2, 4\rightarrow 16}$(poor)                 &1.64E-05  &4.62E-03   &1.4032E-10 & 2.70E-04   \\ \hline
\end{tabular}
\end{table}

Observations from Figures \ref{fig:fig7} and \ref{fig:fig8} also support this analysis. When the dither sequence $m_{1, \cdot \rightarrow \cdot}$, which is in fact 'optimal' in both the MSD and $D_{\text{KL}}$ sense, the prediction obtained from the post-experimental analysis accurately fits the results from the desired reference signal for both ZF and MMSE structures. However, when the 'poor' mapping, $m_{2, \cdot \rightarrow \cdot}$ was used, the BER prediction of MMSE-DFE started to deviate from the reference, while the prediction for the ZF-DFE remained valid. 

\begin{figure}[H]
	\centering
	\includegraphics[scale=0.8]{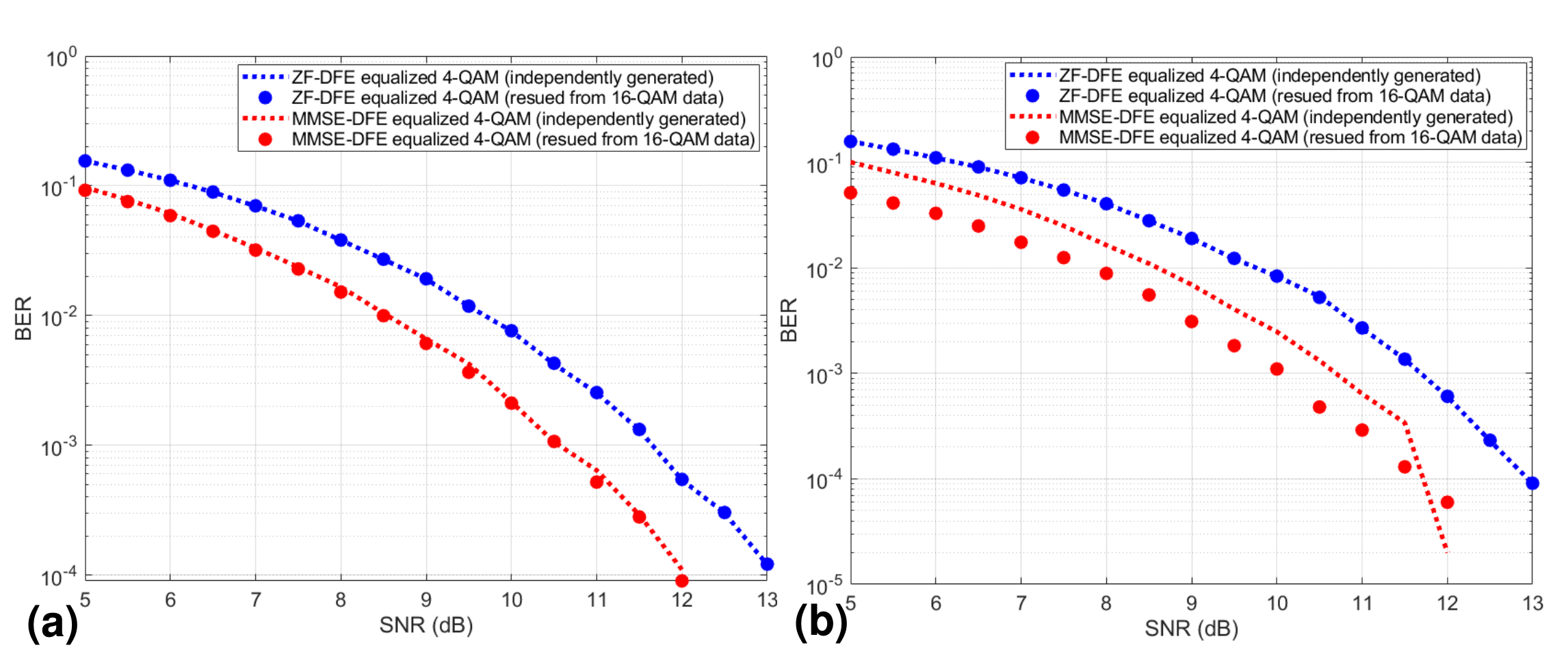}
	\caption{Comparison of transmit SNR vs. BER between independently generated QPSK signals and one reproduced from 16QAM data; two different dither sequences, (a) $m_{1,16\rightarrow 4}$ (optimal) and (b) $m_{2,16\rightarrow 4}$ (poor) were each inserted to reproduce QPSK signals from 16QAM data; ZF-DFE (blue) and MMSE-DFE (red) were used to equalize intersymbol interference. \label{fig:fig7}}
\end{figure}

\begin{figure}[H]
	\centering
	\includegraphics[scale=0.8]{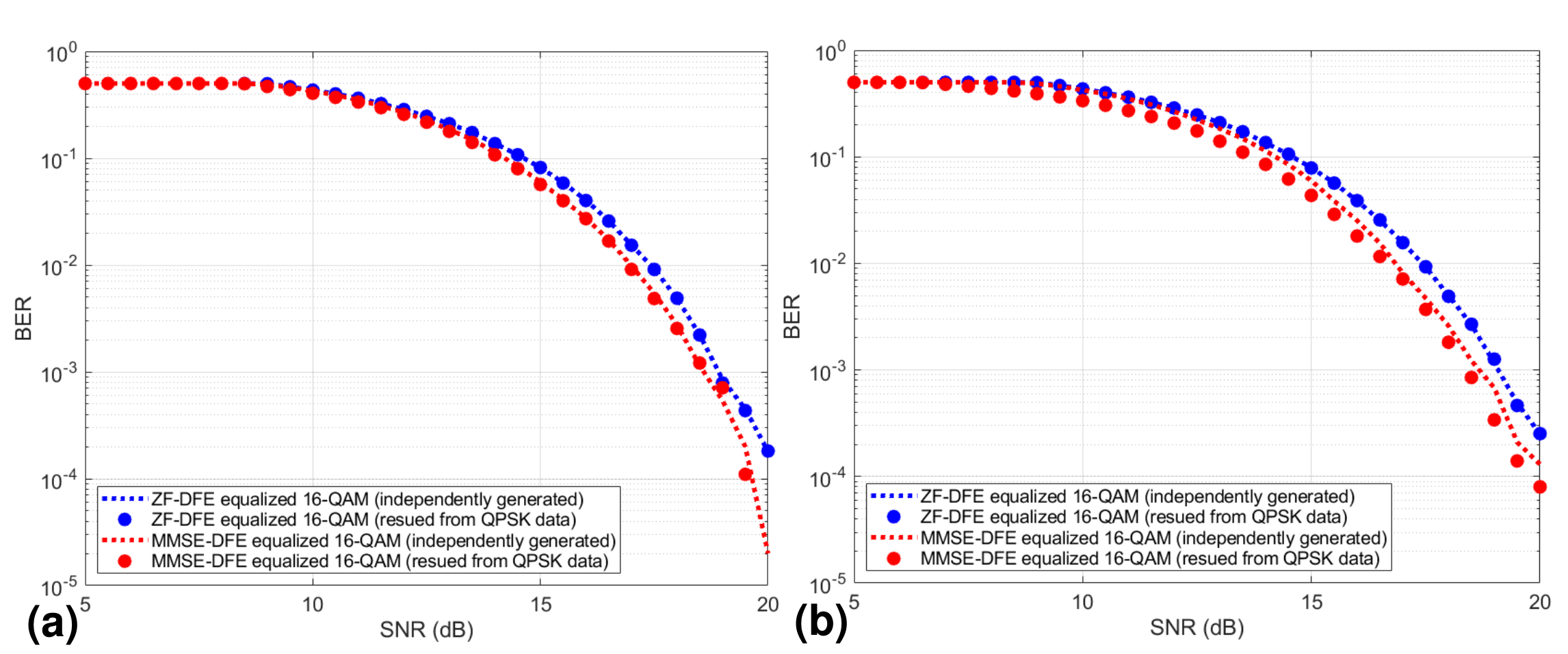}
	\caption{Comparison of transmit SNR vs. BER between independently generated 16QAM signals and one reproduced from QPSK data; two different dither sequences, (a) $m_{1,4\rightarrow 16} $ (optimal) and (b) $m_{2,4\rightarrow 16}$ (poor) were each inserted to reproduce QPSK signals from 16QAM data; ZF-DFE (blue) and MMSE-DFE (red) were used to equalize inter symbol interference.\label{fig:fig8}}
\end{figure}

\subsection{Adaptive Implementations}
While the above analysis assumed a fixed receiver, most underwater acoustic communication systems \cite{riedl13} use a form of adaptive equalization. In this section, we show that the dither insertion approach can be used in adaptive systems.

To illustrate, an LMS-DFE was implemented for a data set for which the first 3,000 transmit symbols were used as training data. As in previous simulations, by reusing 20dB SNR 16QAM data sent over the synthetic channel, using dither $m_{1, 4\rightarrow 16}$, we emulate the output of 20dB QPSK signal transmission at an LMS-DFE receiver. A standard LMS algorithm with two different step sizes $\mu=0.001, 0.005$, were chosen. It should be noted that the error symbols used in LMS updates are computed through dither addition, while the input to the filter remains immutable.  For comparison, 20dB QPSK transmissions without dither were also performed as a reference.
Not only do the outputs after convergence track the reference output, as illustrated in Figure \ref{fig:LMSLMS} (a), but it also simulates the convergence behavior of the adaptive algorithm; two different rates were well-predicted using the proposed dither method.

\begin{figure}[h]
	\centering
	\includegraphics[scale=0.5]{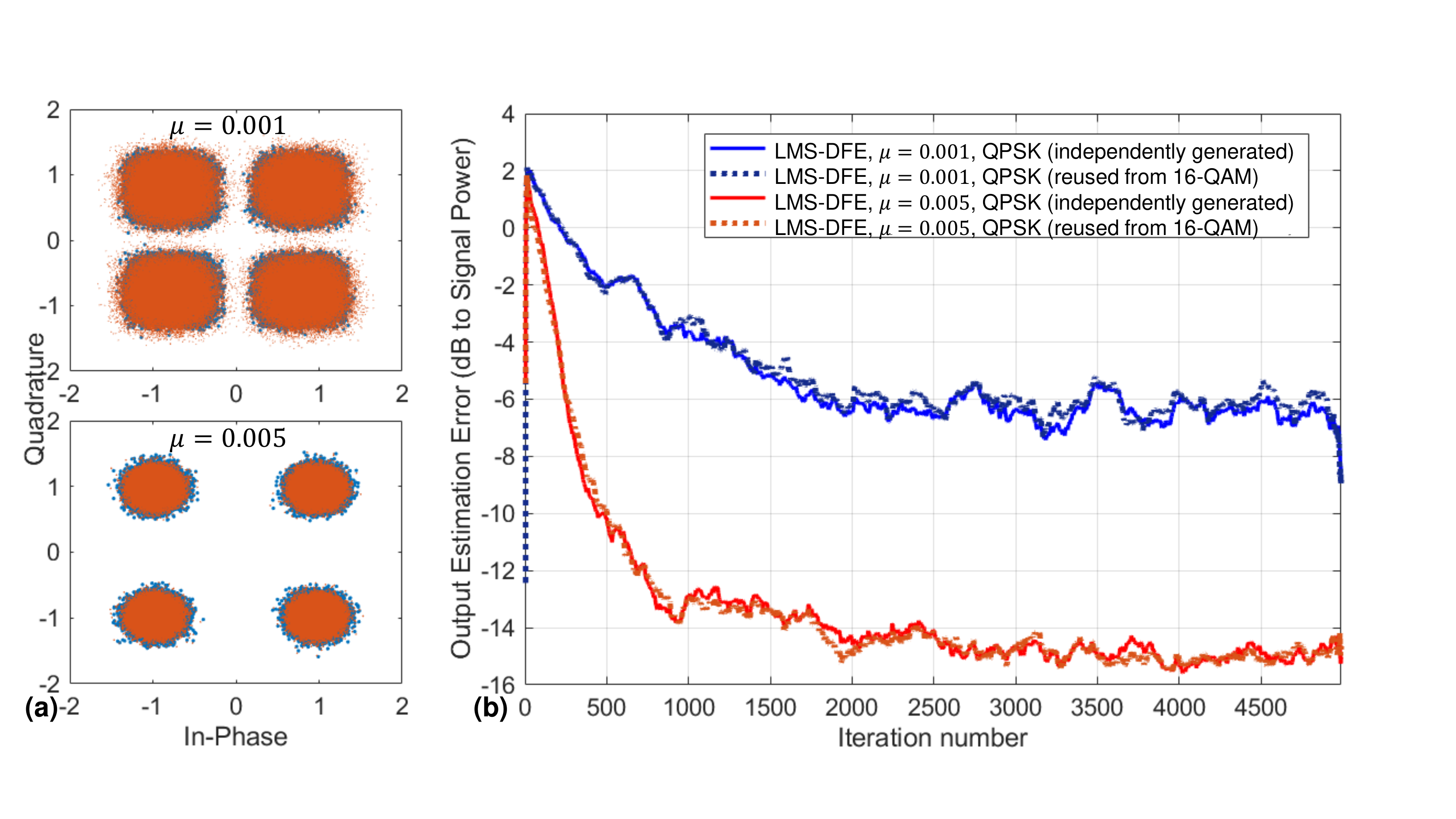}
	\caption{(a) Output constellation for independently-generated QPSK transmission (red) and for dither-generated 16QAM transmission. (b) Learning curves for LMS-DFE with step sizes $\mu=0.001$ (top) and $\mu=0.005$.\label{fig:LMSLMS}}
\end{figure}

\section{Post-experimental Reproduction of MACE 2010 Dataset}

The proposed method using additive dither was tested on experimental data from the Mobile Acoustic Communication Experiment (MACE) conducted in 2010. The experiment was conducted 100km south of Martha's Vineyard, MA. A V-fin with an array of transducers was towed around a ``race track'' configuration of approximate size 600m by 3.8km. The maximum tow speed was around 3kt(1.5m/s) and tow depth varied between 30m and 60m. The signal was transmitted from the towed array and received from a 12 channel hydrophone array moored at a depth of 50m, approximately 5 km away from the ``race track.''

The transmission was single input, multiple output (SIMO) with 1 transmit and 12 receive elements, but only 4 hydrophones at the receiver were used in this example, i.e. $M_r = 4$. Transmit constellations of both QPSK and 16 QAM were used at 9.765625 ksps at a carrier frequency of 13 kHz and receive sample rate of 39.0625 ksps. Channel impairments include time-varying delay and Dopper spread, for which a recursive least square (RLS) DFE was used, with 10000 training symbols for initialization. Detailed explanation of the receiver structure is given in Appendix B. 
\begin{figure}[H]
	\centering
	\includegraphics[scale=0.5]{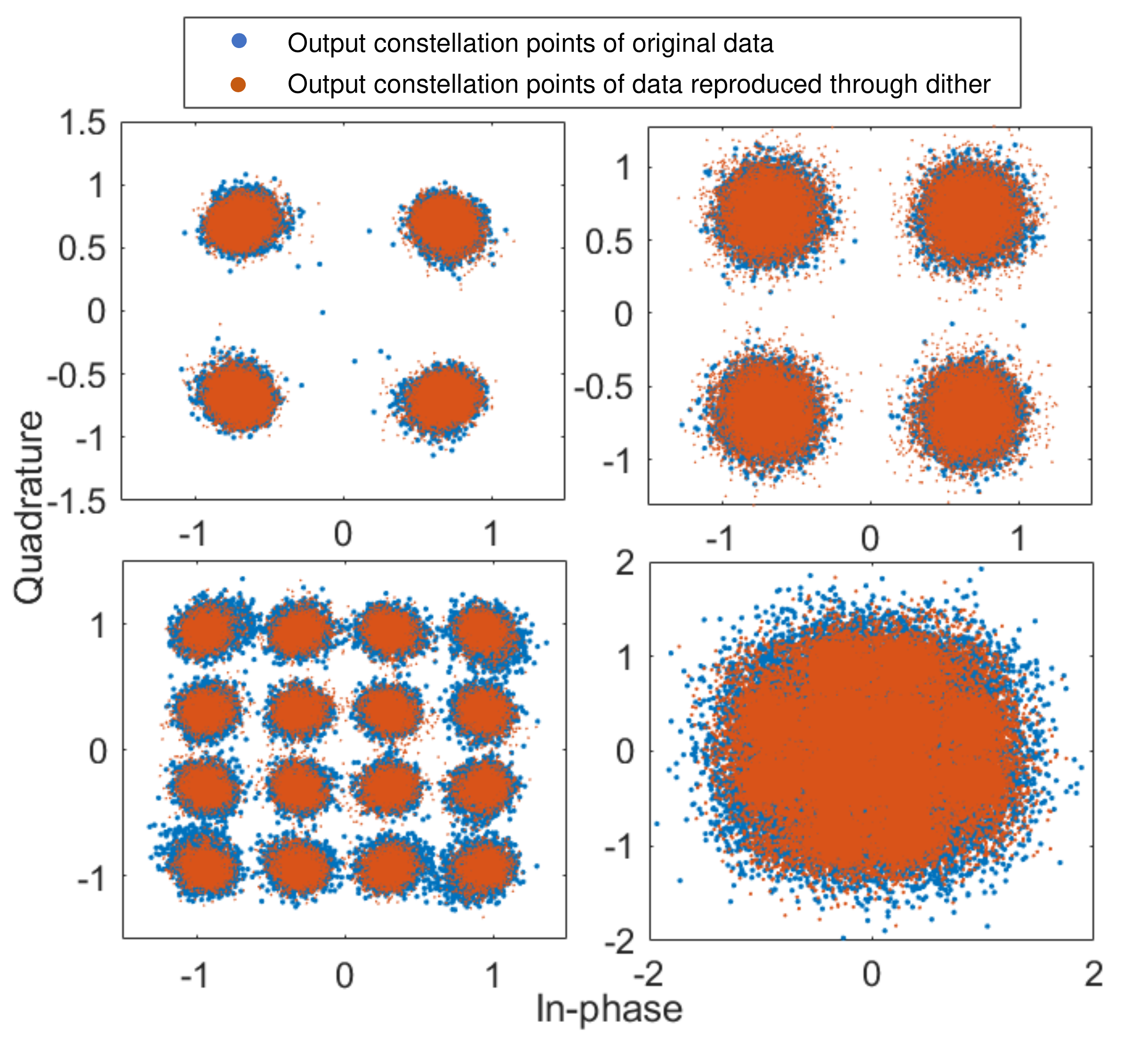}
	\caption{RLS-DFE equalized symbol constellations from the original MACE 2010 data are shown blue, with symbol constellations generated from dither shown in red: (top left) SIMO, $M_r = 4$ QPSK, (top right) SISO QPSK, (bottom left) SIMO, $M_r = 4$ 16QAM, and (bottom right) SISO 16QAM. \label{fig:MACE_process}}
\end{figure}

The purpose of this section is to illustrate the validity of the proposed dither-based post-experimental method, showing that data reproduced through dither insertion closely tracks the behavior of at-sea experimental channels and systems. Among the MACE 2010 data, two sets of transmissions (one with QPSK and one with 16QAM) were used to produce dither simulations. Since the symbol rates were the same for both sets, a QPSK data set was created using a dither from the 16 QAM data, and a 16 QAM data set was created using a dither from the QPSK data.  Both SIMO receivers, using 4 receive hydrophones and SISO receivers, using a single hydrophone, were examined.  These 4 dither-constructed signals used the $m_{1, \cdot \rightarrow \cdot}$ dither. Also, to validate the accuracy of the dither-based method, we compared our results not only to the original experimental data, but also to the data generated from the traditional channel-playback method. To generate the playback output, RLS-based adaptive channel estimation were used to track the time-varying impulse responses, which are modeled in 800 taps in the baseband regime. Also, time-varying Doppler estimates captured from a phased locked loop (PLL) at the receiver were used to simulate the Doppler spreads. Noise recordings measured separately from the same environment were used to estimate the noise power, and ambient noise in the playback simulation was generated as additive and Gaussian. 

\begin{figure}[H]
	\centering
	\includegraphics[scale=0.6]{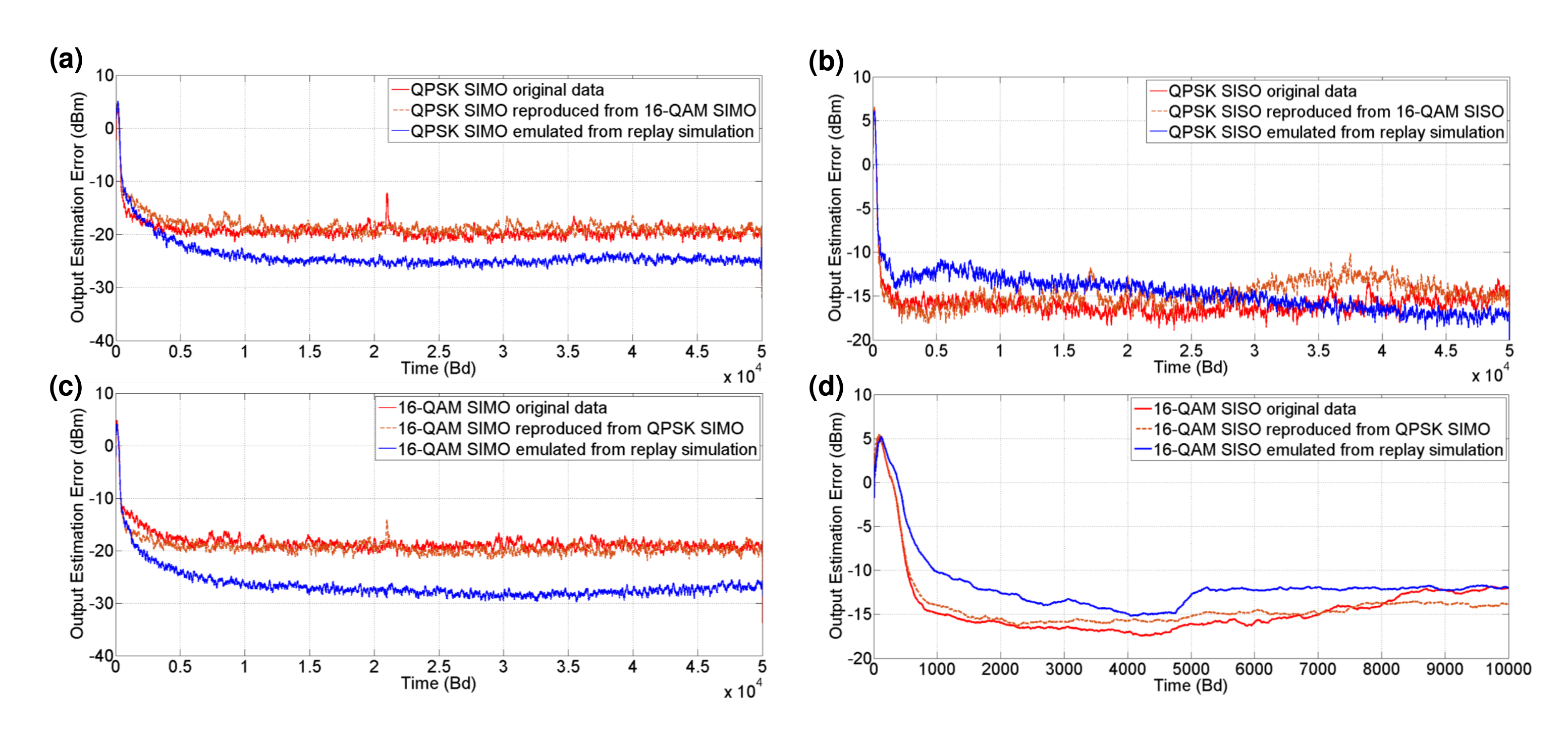}
	\caption{MSE learning rates for the RLS-DFE: (a) Comparison among the original QPSK data, one reproduced through dither, and the replay-simulated data based on channel estimations from 16 QAM recordings, for $1 \times 4$ SIMO receivers and (b) a single SISO receiver. Similarly 16 QAM comparisons are expressed in (c) SIMO and (d) SISO receivers. \label{fig:RLS_converge}}
\end{figure}

Figure \ref{fig:MACE_process} illustrates both SIMO and SISO results, with constellation points well-aligned to the non-dither transmit constellations. One interesting consideration is the case of 16QAM SISO transmission, for which we can see that even the failure of the transmission can be predicted from the proposed method. Note that while this does not imply that the distribution of the output can be accurately matched by dither simulation, the mean squared error can be reasonably-well predicted even after equalizer failure, as shown in Figure \ref{fig:MACE_process}. We also see that the RLS convergence speed (learning curve) is well-predicted by the dither simulation in Figure \ref{fig:RLS_converge}. For example, while the channel state in the QPSK SIMO original data is different from that in the 16QAM SIMO data in Figure \ref{fig:MACE_process} (a), their learning curves largely track one another, since these two data transmissions were conducted within a short period of time, in same location. Once again, even though the SISO 16 QAM data sets failed to converge after training, (and hence are not shown beyond the training period), the leanring curves for both the original 16 QAM data and that for the dither-generated 16 QAM data show similar learning curves through the training stage. While data reproduced from our suggested method show accurate tracking to the original experiments, it can also be observed that performance predictions from the conventional channel playback simulator are often overestimated more than 5 dB, as illustrated in Figure \ref{fig:RLS_converge} (a) and (c). There are a few reasons for this. First, the channel equalization algorithms are inherently smoothing channel fading effects, which may be faster than that estimated from the data. Second, as we explore in the next section, direct channel playback as employed here assumes that the channel estimates are correct, and playback simulations fail to adequately capture unmodeled channel estimation error.

\section{Direct Channel Playback Simulation}

\indent A second approach to using data from at-sea experiments to test new transmit signal modulation and coding techniques and signal processing algorithms is the use of direct replay simulators.  A study of the use of such simulators is contained in \cite{otnes2013validation}.  
In concept, channel probes are transmitted through the ocean channel with the received signal used to estimate the acoustic channel's time-varying impulse response (TVIR). In addition, ambient noise signals are collected at times when the channel probes are not transmitted.  During direct playback simulation, the communications signals to be tested are convolved with the estimates of the TVIR, ambient noise is added to achieve a desired SNR and the resulting signal is processed to evaluate either the candidate transmitted signal or signal processing algorithms.  It is important that the bandwidth of the channel probe signal encompass at least the entire transmission band of the communications signal. The ambient noise can be added either using the directly measured noise or noise generated to have similar statistics, but this is not the focus of this section of the paper.

A challenge of direct playback simulation arises from for the time-variability of the channel.  Here, the discussion focuses on communications paradigms using channel equalization. Specifically, compensating for channel dynamics is one of the key challenges in signal demodulation, whether or not it involves directly estimating the TVIR.  Thus, it is important that the dynamics of the channel used in direct playback simulation represent the channel dynamics in the actual channel.  However, channel dynamics will also impact the efficacy of the estimated TVIR with more rapidly-varying channel features being more difficult to estimate accurately than more slowly-varying features.  Thus, the very features of the channel that impede reliable communications may not be captured in the estimates of the TVIR.  The approach proposed here is intended to address this challenge and make direct playback simulation more useful in rapidly varying conditions.

The approach here uses the error metric created by a channel estimation algorithm to estimate and account for the impact of channel estimation errors in a channel playback simulator. By subtracting a prediction of the received signal, based on the current channel estimate and the transmitted signal, from the actual received signal, the residual prediction error (RPE) provides a convenient measure of the impact of the excess channel modeling error on performance prediction based on channel replay simulation. 

The channel model is that of the complex baseband channel between the transmitted sequence ($g[n]$ in Figure 1) and the sampled complex baseband received signal ($y[n]$ in Figure 2).  This specifically incorporates the transmit pulse shaping and modulation to passband of the channel probe signal into the channel model.  These may be aspects of the signal generation that a user would like to investigate.  In this case, the extension of the proposed approach to provide this flexibility is straightforward but is not addressed herein. The baseband channel model is given by
 
\begin{equation}
    y_k[n] = \sum_{l=0}^{L-1}h_{k}^{*}[n,l]g[n-l]+v[n],
\end{equation}
  
where $k$,$k=1,⋯,K$ is the hydrophone number for a $K$-element receiver array, $L$ is the delay spread of the TVIR, the superscript $*$ denotes complex conjugate and $v[n]$ is the complex-baseband received noise.  Based upon this model, the estimate of the complex baseband TVIR, denoted as $\hat{h}_k[n,l]$, is estimated as the solution to the least-squares problem
 
\begin{equation}
    \hat{h}_k[n,l] = \argmin_{h_k}\sum_{n'\leq n}\bigg|{y_k[n']-\sum_{l'=0}^{L-1}h_{k}^{*}[n',l']g[n'-l']}\bigg|^2.
\end{equation}
  
Here, the summation of $n'$ is over the averaging window of the least-squares estimator.  There are many modifications of this estimation approach including exponential windowing and the incorporation of explicit models of the channel time variability to improve the algorithm’s ability to estimate a TVIR. These algorithms can be used in place of the simple algorithm shown here without changing the validity of the direct playback simulation approach proposed here.

Define the {\it{a priori}} signal estimate as $\hat{y}_k[n|n-1] \triangleq \sum_{l=0}^{L-1}h_{k}^{*}[n-1,l]g[n-l]$  and the least-squares residual signal prediction error as $e_k[n|n-1]\triangleq y_k[n]-\hat{y}_k[n|n-1]$. The notation $[n|n-1]$ indicates that the TVIR estimate from time $n-1$ is used to predict the received signal at time $n$.  Define the corresponding {\it{a priori channel}} estimation error as $\epsilon_k[n,l]\triangleq h_k[n,l]- \hat{h}_k[n-1,l]$. Then the true received signal can be written as
 
\begin{equation}
    y_k[n] = \hat{y}_k[n|n-1]+\sum_{l=0}^{L-1}\epsilon_{k}^{*}[n,l]g[n-l]+v[n].
\end{equation}

The first term is the output of the estimated channel given the transmitted signal, $g[n]$, the second term is the additional portion of the received signal resulting from channel estimation errors and the final term is the observation noise.  The last two terms comprise $e_k[n|n-1]$.  Under the mild assumption that the channel estimate is an MMSE estimate, the second term is independent from the first. This residual prediction error can be calculated directly from the channel estimation algorithm.  

In \cite{preisig2005performance}
it was shown that appropriately characterizing $e_k[n|n-1]$ and incorporating this in the calculation of channel estimate-based equalizer filter weights improves the performance of the resulting equalizer.  Here it is proposed to use the characterization of $e_k[n|n-1]$ to enable direct channel playback simulators to appropriately compensate for channel estimation errors. When the channel probe signal is received at a reasonably high SNR, $e_k[n|n-1]$ should be dominated by the channel estimation error induced signal rather than the observation noise.  This is assumed here.

Let $\boldsymbol{e}[n|n-1]=[e_1[n|n-1],\cdots ,e_K[n|n-1]]^T$ and similarly define $\hat{\boldsymbol{y}}[n|n-1]$ and $\boldsymbol{y}[n]$. From the channel estimation algorithm, calculate $\boldsymbol{e}[n|n-1]=\boldsymbol{y}[n]-\hat{\boldsymbol{y}}[n|n-1]$.  From the time series of $\boldsymbol{e}[n|n-1]$, estimate the spatial-temporal correlation matrix given by $\boldsymbol{R}_e[m]\triangleq E[\boldsymbol{e}[n|n-1] \boldsymbol{e}^H [n-m|n-1-m]]$, where the superscript $H$ denotes Hermitian and temporal stationarity of the time series is assumed.  For these data, calculate the variance of the transmitted sequence, $g[n]$, denoted $\sigma_g^2$. For simplicity, we assume that if the transmit array has more than 1 element ($M>1$), then the channel probe signal is uncorrelated from element-to-element and has the same variance at each element.  It is straightforward to modify the approach if this is not the case.

Denote the transmitted communications signal to be used in the direct playback simulator by $\tilde{g}[n]$ and denote its variance by $\tilde{\sigma}_g^2$. First, generate a complex normal time-series, ${\tilde{\boldsymbol{e}}}[n]$, with spatial-temporal correlation given by $(\tilde{\sigma}_g^2/\sigma_g^2)\boldsymbol{R}_e[m]$ .  Finally, generate additive noise, $\tilde{\boldsymbol{v}}[n]$, using either prerecorded noise segments or via statistics that match the recorded ambient noise as would be done in a standard direct playback simulator.  Finally, generate the simulated received time series as

\begin{equation}
    \boldsymbol{\tilde{y}} = \sum_{l=0}^{L-1}\hat{\boldsymbol{h}}^{*}[n-1,l]\tilde{g}[n-l]+\tilde{\boldsymbol{e}}[n]+\tilde{\boldsymbol{v}}[n].
\end{equation}
  
As constructed here, the simulator output consists of a sum of the test signal propagated through the estimated channel, the modeled or recorded observation noise, and an additional residual prediction error (RPE) term, that is the part of the received signal corresponding to difference between signal propagating through the true channel and the estimated channel. The use of ``RPE'' in the simulator is the proposed modification to direct playback simulation to account for channel estimation errors.

\subsection{Experimental Data from KAM11}

Shown in Fig.~\ref{fig:KAM11}, are data collected from the KAM11 data collection experiment. On the left, is the estimated channel impulse response over time, illustrated in dB relative to the signal peak at each time slot.  The received passband communication signal power versus time and the pass band received noise power (measured during a silence interval) are also shown on a dB scale. From the center (communication signal power) and right (noise power) figures, it can be seen that the SNR is in excess of 25dB for this measurement. It is also evident from the noise power estimates, that there is frequent impulsive noise in the environment. This data set was collected with a 4-element vertical array, with a 10cm element spacing. The transmit center frequency was 13 kHz, and data was transmitted at 6250 symbols per second, using a BPSK modulated m-sequence continuously repeated for 50 seconds.  The source to receiver range was 3 km in 100 m depth of water with the source and receiver at 15 m depth, enabling surface dynamics to play a role in the channel structure. For this set of channel measurements, during direct channel playback simulation, the residual prediction error (RPE) is due primarily to errors in estimating channel impulse response.

\begin{figure}[h]
	\centering
	\includegraphics[scale=0.8]{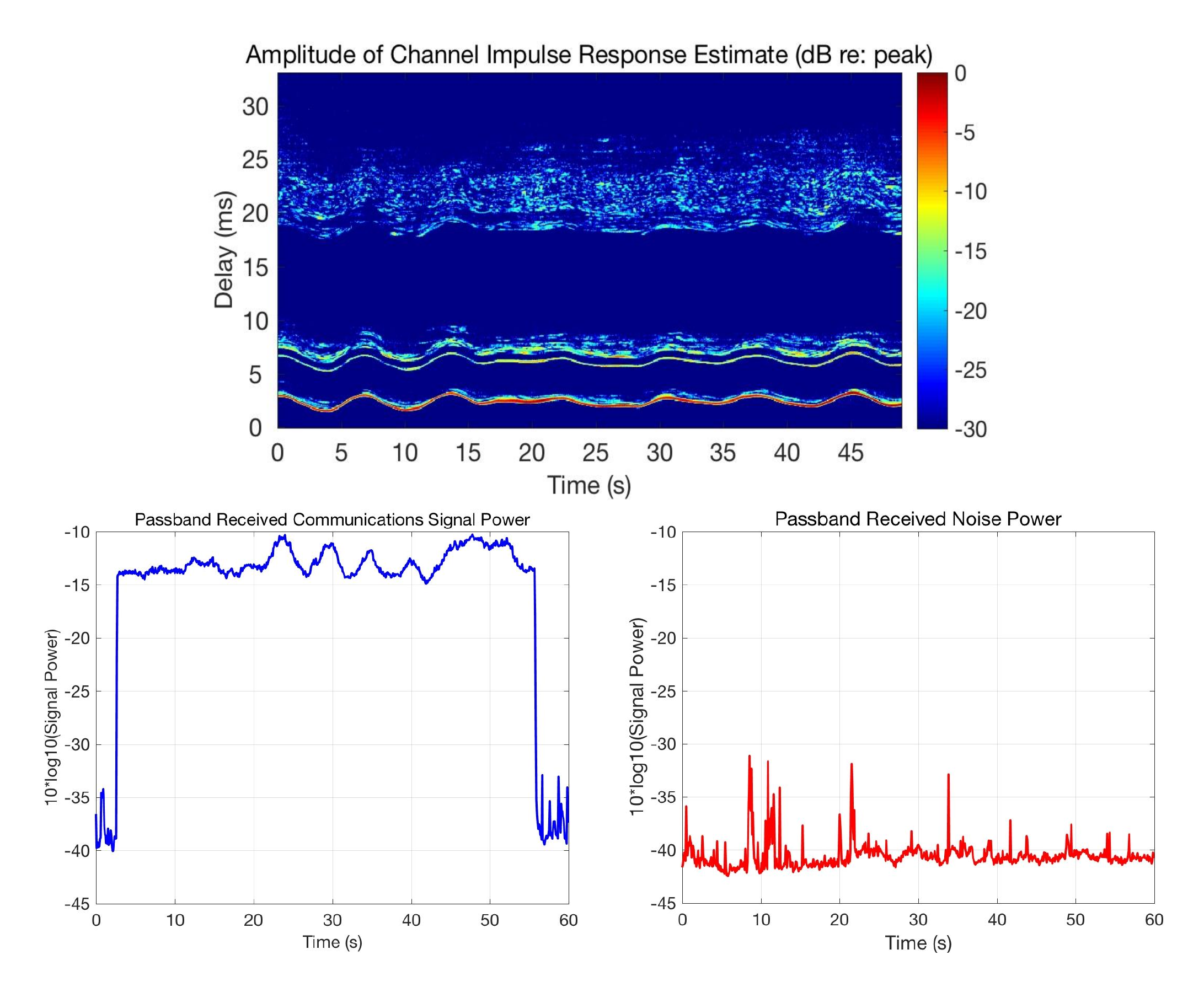}
	\caption{KAM11 passband communication data set collected in 100 m depth with source and receiver at 15 m from the surface, separated by 3km. Shown are (left) impulse response measurements, (center) received passband signal power, and (right) passband noise power.\label{fig:KAM11}}
\end{figure}

Figure \ref{fig:DCPRPEpower} illustrates a comparison between the performance predictions of direct channel playback with and without additional residual prediction error in the simulation.  As can be evidenced by the lower two lines in the figure, the residual prediction error of the channel estimator is at the same power level as the noise-free output of the direct channel playback simulator. The upper two plots illustrate the measured communication signal power along with that from direct channel playback with RPE added to the output. The dominant source of error between noise-free direct channel playback and the received signal power aligns well with the measured RPE, indicating that RPE indeed dominates system performance predictions for this data set.

\begin{figure}[h]
	\centering
	\includegraphics[scale=0.6]{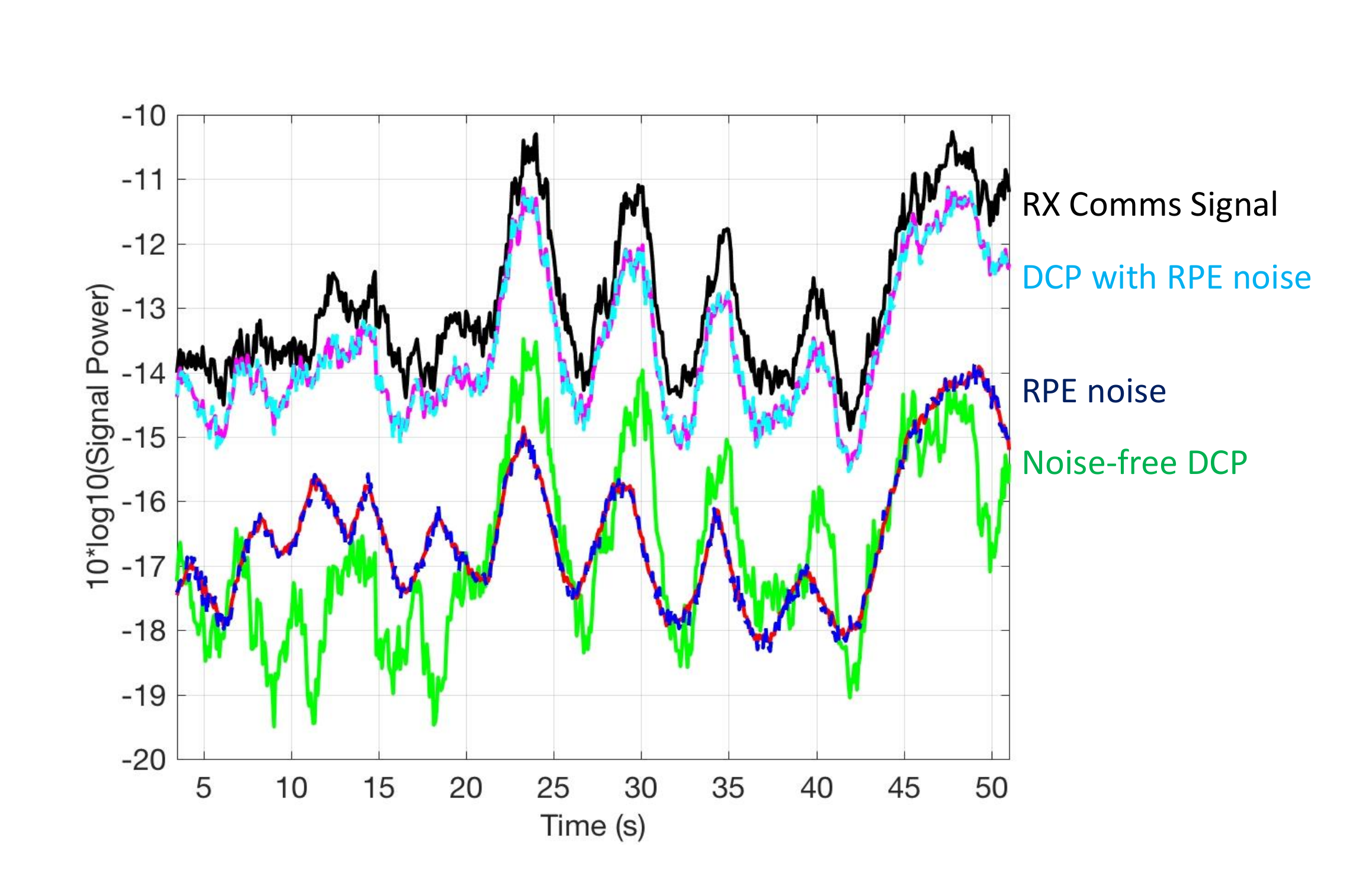}
	\caption{Output signal power versus time for (top) the received communication signal, (second) direct channel playback with RPE noise added, (middle) RPE noise vs. time, and (bottom) noise-fre direct channel playback.\label{fig:DCPRPEpower}}
\end{figure}

Figure ~\ref{fig:EQPerfRPE} illustrates  the soft-decision SNR performance (taken at the slicer output) in dB of an equalizer operating on the signals whose power were shown in Fig.~\ref{fig:DCPRPEpower}, for two representative synchronization sequences. the first (left) illustrates the gap between the performance predicted by DCP alone and that achieved in the field. Note that the DCP output over-predicts performance in terms of SNR and that this over-prediction is worse at higher SNR, where RPE is on the order of the signal received signal power. The second example (right) shows performance fading in and out and it is notable that the performance of the system tracks its estimated performance well for the DCP with RPE both during periods of strong and poor performance, while that of DCP alone only accurately predicts performance when it is poor, during signal fades.
\begin{figure}[h]
	\centering
	\includegraphics[scale=0.55]{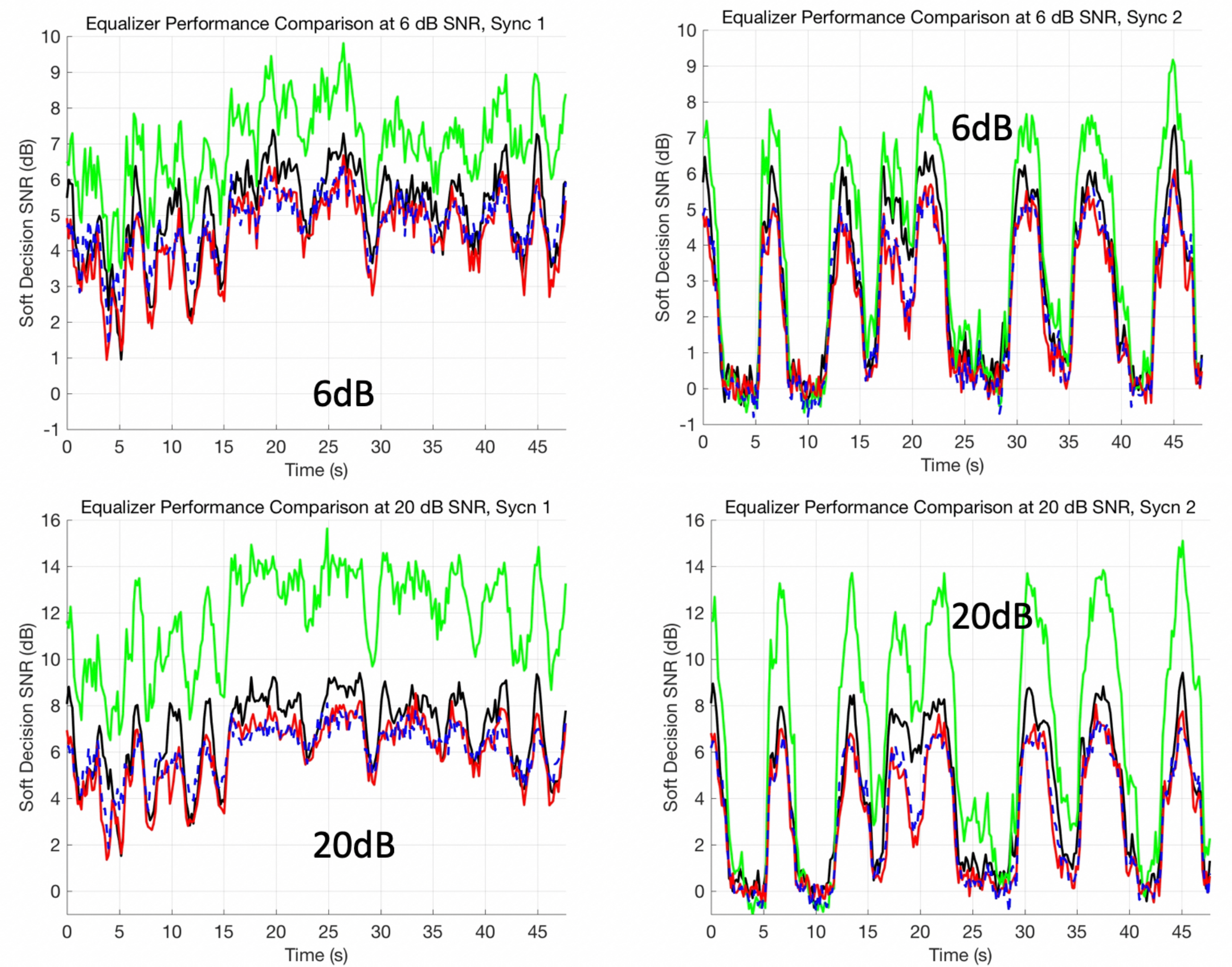}
	\caption{Shown are the soft-decision SNR performance for equalizer outputs of the received transmission, and that of DCP and DCP with RPE for two different transmitted signals at each of two SNRs.
\label{fig:EQPerfRPE}}
\end{figure}

An important element of this approach is the ability to augment channel replay simulation with an additional noise term based on residual prediction error. The goal of this approach is to provide more accurate performance predictions from channel replay simulation than would be achieved through simple time-varying convolution with an estimated channel response, which would provide an overly optimistic estimate of performance, since it takes as ``ground truth'' the estimated channel, ignoring the unmodeled dynamics. If the ambient environmental noise is non-Gaussian, then the residual prediction error will have channel-estimation induced components in addition to non-Gaussian components, at modest SNR. As such, simply estimating second-order statistics of the RPE will no longer be sufficient for determining the error power due to the transmitted signal convolved with the channel estimation error.  While at high SNR, RPE will still be dominated by channel estimation error induced Gaussian components, at modest SNRs, the non-Gaussian components of RPE need to be separated from the Gaussian components and accounted for appropriately, together with the non-Gaussian environmental noise components that are added to achieve a given SNR.
While the purpose of channel replay simulators is to develop performance predictions for new algorithms or signaling strategies, estimates of performance that make use of second-order statistics of the signal and noise (SNR or MSE) based on linear processing and Gaussian statistics no longer hold. For non-Gaussian environmental and RPE noise terms, it will be more difficult to predict the impact on performance from this excess unmodeled component using synthetic additive noise terms, and a simple additive RPE term may not be sufficient. Such cases will need to make use of ambient noise recordings, or synthetic noise sources that better fit the true environmental statistics.

\section{Environmental Models}

\begin{figure}[h]
	\centering
	\includegraphics[scale=0.5]{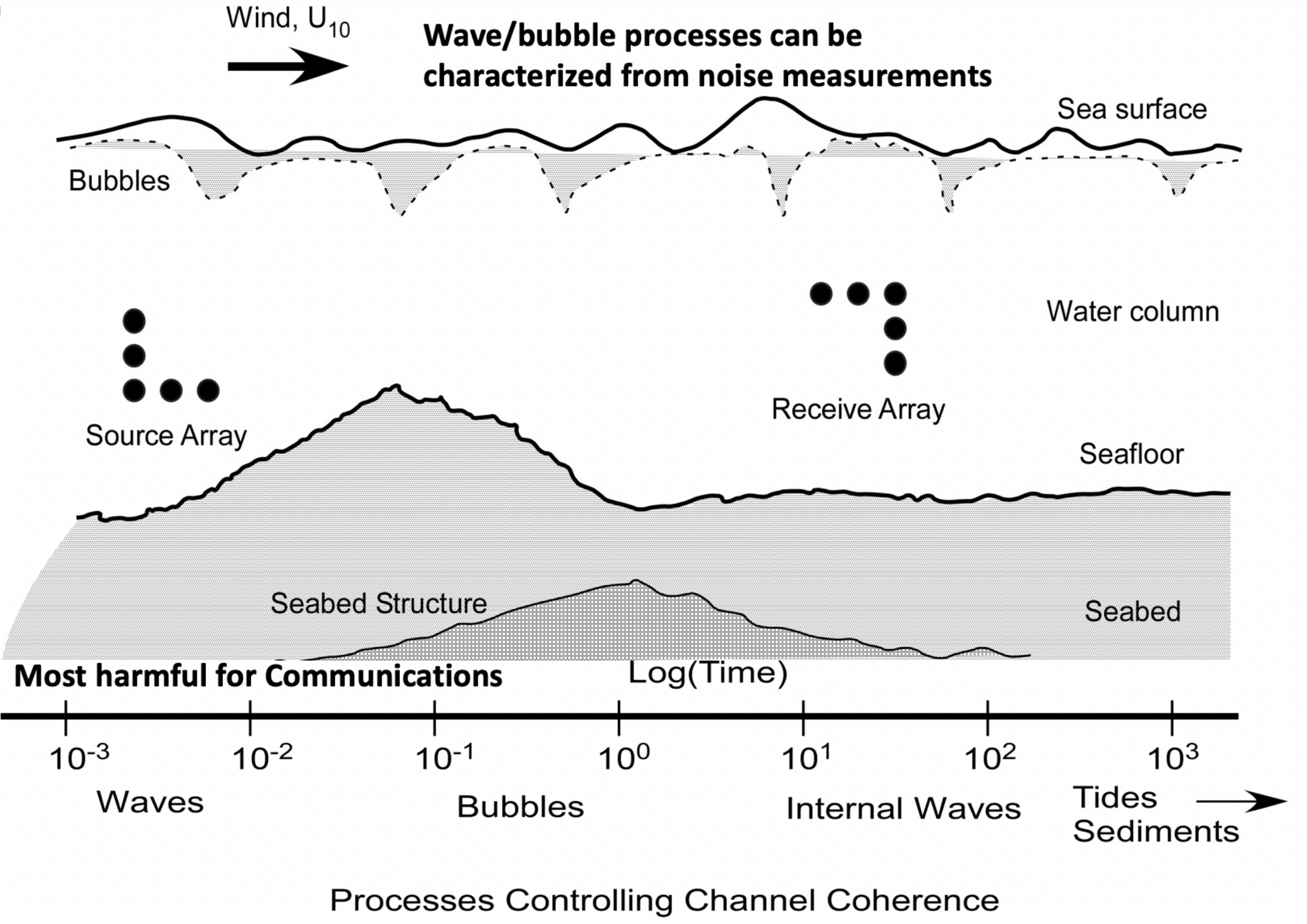}
	\caption{Channel and noise coherence drivers in underwater acoustic communications.\label{fig:CohereneDrivers}}
\end{figure}
As illustrated in Figure \ref{fig:CohereneDrivers}, there are a variety of factors that drive the coherence of channel and noise statistics in the underwater acoustic environment and these span timescales over several orders of magnitude. Often, these environmental parameters induce non-white and non-Gaussian noise components. For example, in \cite{traverso2012simulation} it was reported that the additive non-Gaussian ambient noise is induced by both strong winds and shipping movements, and can be modeled parametrically by these environmental factors. 
In channel replay simulation with non-Gaussian ambient noise, realistic performance prediction and simulation can only be achieved if environmental factors governing the non-Gaussian statics can be estimated accurately.

In this context, we now discuss methods to estimate environmental parameters, such as wind speed and seafloor type, that determine channel propagation regimes.  The surface creates significant dynamic variability in the ocean acoustic communications channel.  Gravity waves, driven by local ocean winds or distant storms, scatter and focus surface-reflected energy and introduce second-order terms into intensity and Doppler statistics \cite{preisig2004surface}.  Moreover, bubbles are episodically and randomly entrained by breaking waves, which both radiate acoustic noise and screen the surface acoustically through sound absorption \cite{deane2013suspension}.

The entrainment of bubbles near the sea surface creates channel regimes, the properties of which depend on wind speed.  Few bubbles are entrained at wind speeds in the range 0-3 m/s, and surface-reflected energy contains high-energy, Doppler-shifter arrivals \cite{preisig2004surface}.  Wind speeds higher than 3 m/s drive bubble entrainment through wave breaking, leading to increased scattering and absorption and less surface-scattered energy.  A second regime shift occurs at wind speeds around 13 m/s, which generates enough turbulence in the near-surface boundary layer to suspend bubbles across a broad range of sizes, resulting in a transition from mid frequency to low frequency surface screening \cite{deane2013suspension}.

The direct link between wind and bubble entrainment can be exploited to estimate wind speed from ambient noise, in the absence of shipping, rain and biological sounds (see \cite{vagle1990evaluation} for methods to detect noise segments that are free of contaminating sources).  This technique, called Weather Observations through Ambient Noise or WOTAN, is described in detail in \cite{vagle1990evaluation} and provides estimates of wind speed accurate to $O(1)$ ${\text{ms}^{-1}}$ for wind speeds less than 16 ${\text{ms}^{-1}}$.  One advantage of the WOTAN algorithm is that it can be implemented with ambient noise observed on a single hydrophone between signal transmissions and thus exploits hardware which must be present for a functional communications system. The wind speed at 10 m height, $\text{U}_{10}$, is calculated as follows. Let SSL be the measured sound spectral level averaged over an hour in dB re 1$\mu \text{Pa}^2$/Hz at the frequency f kHz, Q = -19 dB/decade is the noise spectrum slope and $\beta$ is a frequency-dependent factor accounting for attenuation and refractions that lies in the nominal range 0-1 dB for hydrophone depths less than 150 m and frequencies below 25 kHz. Then $\text{U}_{10}=(p_0+80.94)/52.8$, where $p_0 = 10^{SSL_0/20}$ and $SSL_0 = SSL(f)+Q\log(f_0/f)+\beta$ where $f_0 = 8$ kHz is the reference frequency.

If a vertical hydrophone array is available, information about the geophysical properties of the seafloor can also be obtained from the ambient noise field.  This technique, known as passive fathometer processing, is described in \cite{siderius2010adaptive} and \cite{harrison2002geoacoustic}. Seafloor reflectivity can be recovered from a few minutes of noise data \cite{harrison2004sub}. Given the wealth of information available in the ambient noise field, recording segments of noise data along with acoustic transmissions to provide information about channel state is strongly encouraged. Post-processing techniques \cite{vagle1990evaluation} may need to be applied to test for the presence of noise sources other than breaking waves.

\section{Conclusions}

This paper discusses useful steps that can be taken and best practices that have proven useful to maximize the value of field experiments.  These include preparation of signals for transmission and their collection such that research trades for different modulation and coding schemes may be undertaken post-experiment, without the need for retransmission of additional waveforms. Coding strategies were explored through adding a whitening sequence and an XOR in the bitpath and can be applied to any field data. Modulation strategies were explored by inserting dither sequences in the signal path to map signals onto existing field data. End-to-end channel linearity and effective ISI mitigation through channel equalization enabled system performance prediction, even without channel estimation. Additionally, sufficiently meaningful channel models and noise statistics estimates can enable post-experimental replay of the environment that is more predictive of actual field performance. For replay channel simulation, adding residual prediction errors to the simulation output enabled more accurate performance prediction than conventional replay methods. Finally, collecting sufficient environmental statistics enables meaningful performance comparisons with other field experiments. Collecting experimental data is an expensive endeavor yet it remains fundamental for exploring the limits of the underwater acoustic systems. The data reuse techniques suggested here can be broadly used to make more effective use of collected field data, reducing unnecessary experiments that might attempt to exhaustively explore all communication parameters of interest, while enabling experiment resources to be used judiciously in areas and directions of greatest impact.

\appendices
\section{Derivation of $\mu_{1,k}$, $\mu_{2,k}$, $\sigma_{1,k}$, and $\sigma_{2,k}$}
First, $\mu_{2,k}$ can be derived directly from \eqref{eq:fohatcondition} as such, 
 
\begin{align} 
\begin{split}
    E(\hat{f}_o[n-n_0]|{f[n-n_0]=k}) &= q_f^{*}\boldsymbol{H}_{n_0+1} k + E(q_f^{*} \boldsymbol{H}^{n-n_0} f^{n-n_0}-q_b^{*} f_{n-n_0-1}^{n-n_0-L_b} + q_f^{*}w_{n}^{n-L_f+1}) \\
    &= q_f^{*}\boldsymbol{H}_{n_0+1}k.
\end{split}
\end{align}
  
The second equality comes from the i.i.d. mean zero assumption on $f$ and $w$. Also, $\mu_{1,k}$ can be directly derived from \eqref{eq:fohatcondition} and (23). $\sigma_{2,k}^2$ can be computed as follows:
 
\begin{align} 
\begin{split}
        \sigma_{2,k}^2 &= {\rm Var}(\hat{f}_o[n]-f[n]|f[n]=k) \\
        &= E(|\hat{f}_o[n]-k|^2)- |\mu_{2,k}-k|^2 \\
        &= E(|f[n]|^2) \bigg(q_f^{*}\big(\boldsymbol{H}\boldsymbol{H}+E(|w[n]|^2)/E(|f[n]|^2)\big)^*q_f+q_b^{*}q_b+1\bigg) \\
        &\,\, -2 \Re\{(q_f^{*}\boldsymbol{H}_1^{L_b-n_0-1}{q_b}_{n_0+2}^{L_b}-q_f^{*}\boldsymbol{H}_{n-n_0})\}.
        \end{split}
\end{align}
  
Finally, $\sigma_{1,k}^2$ can be computed similarly from \eqref{eq:10} and \eqref{eq:12} as
 
    \begin{align} 
    \begin{split}
        \sigma_{1,k}^2 &= {\rm Var}(\hat{f}[n]-f[n]|f[n]=k) \\
        &= E(|\hat{f}[n]-k|^2)- |\mu_{1,k}-k|^2 \\
        &= E(|q_f^{*} \boldsymbol{H} g_{n}^{n-L_f-L+2}-q_b^{*} g_{n-n_0-1}^{n-n_0-L_b} + q_f^{*}w_{n}^{n-L_f+1}-g[n-n_0]|^2)-|q_f^* \boldsymbol{H}_{n_0+1}E_{g,k}-g[n-n_0]+m[n-n_0]|^2 \\
        &= E(|g[n]|^2) \bigg(q_f^{*}\big(\boldsymbol{H}\boldsymbol{H}+E(|w[n]|^2)/E(|g[n]|^2)\big)^*q_f+q_b^{*}q_b+1\bigg) -|(q_f^* \boldsymbol{H}_{n_0+1}-1)E_{g,k}|^2 \\
        & \,\, -2\Re\{((q_f^* \boldsymbol{H}_{n_0+1}-1)E_{g,k}E_{m,k}^*)\}-|E_{m,k}|^2.
    \end{split} 
    \end{align}
  
\section{Experimental Setup for MACE data}

\begin{figure}[h]
	\centering
	\includegraphics[scale=0.5]{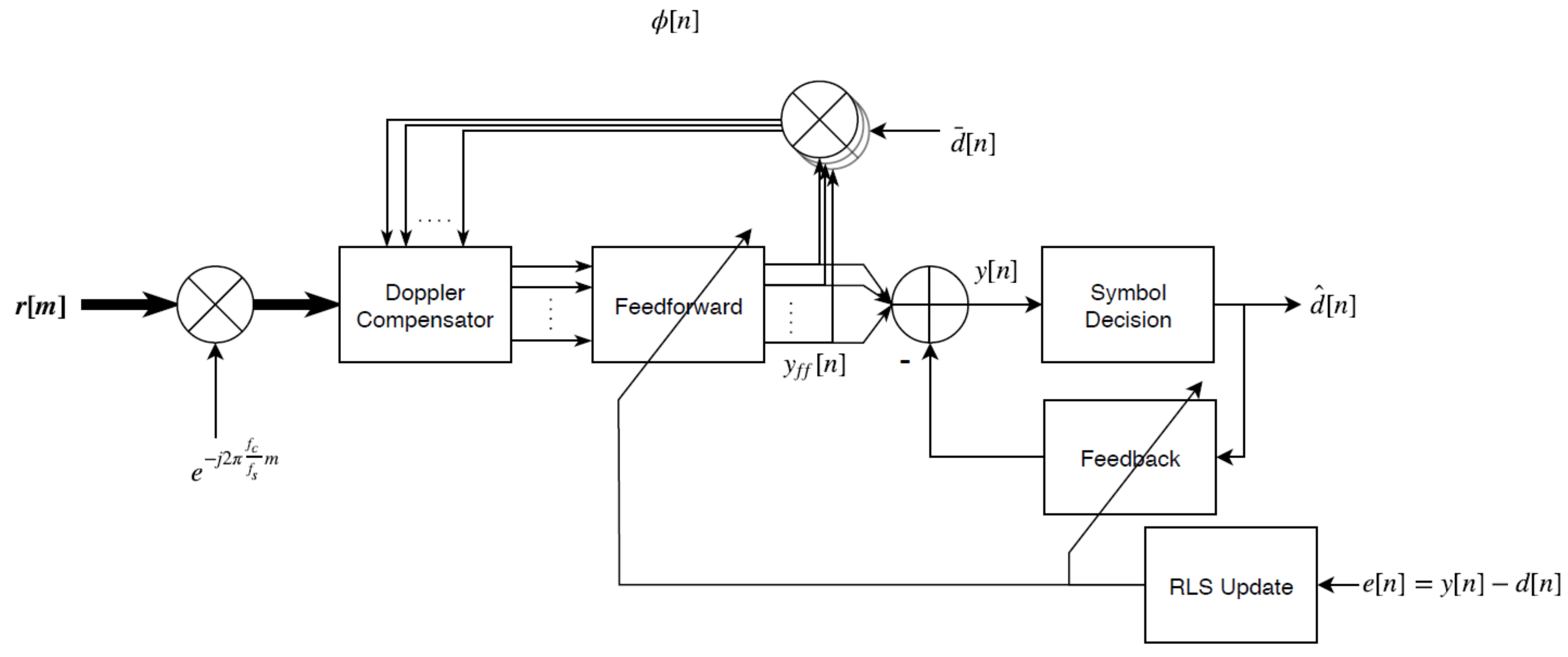}
	\caption{A schematic of the receiver structure used in processing MACE data\label{fig:fig_appendix}}
\end{figure}

The receiver processes each data sequence with a Doppler compensator and a recursive least squares decision feedback equalizer (RLS-DFE). The system design is illustrated in Figure \ref{fig:fig_appendix}. First, 200 symbols are used as a preamble sequence to synchronize the receiver and estimate the initial Doppler scaling. After synchronization, the Doppler for each channel is tracked recursively by employing a phase locked loop (PLL) on the corresponding channel feed forward output. Then, Doppler effects are compensated by resampling and phase shifting the baseband signal. The output of the Doppler compensator is then decoded with a fractionally-spaced RLS-DFE. The reference seqeunce $d[n]$ in the computation of the error sequence $e[n]$ and phase $\phi[n]$ is given during the training period. After training, no further training symbols are used and the equalizer operates in decision-directed mode. 

\section*{Acknowledgment}
\ifCLASSOPTIONcaptionsoff
  \newpage
\fi



\bibliographystyle{IEEEtran}
\bibliography{IEEEabrv,thisisrefs}
%

%




\end{document}